\documentclass[11pt,a4paper]{article}
\usepackage{mathrsfs}

\usepackage{graphicx}
\usepackage{amsmath}
\usepackage{amsthm}
\usepackage{amssymb}
\usepackage{amsfonts}
\usepackage{epsfig}
\usepackage{hyperref}
\usepackage{color}
\usepackage{appendix}
\usepackage{authblk}
\def\Im {\mathop{\rm Im}\nolimits}
\def\arg {\mathop{\rm arg}\nolimits}
\def\Re {\mathop{\rm Re}\nolimits}
\def\Res {\mathop{\rm Res}\nolimits}

\newtheorem{pro}{Proposition}

\newtheorem{thm}{Theorem}[section]
\theoremstyle{remark}
\newtheorem{rem}[thm]{Remark}
\newtheorem{rhp}[thm]{RH problem}

\numberwithin{equation}{section}

\begin{document}
\title{Asymptotics of the  partition function of the  perturbed Gross-Witten-Wadia unitary matrix model}

\author[a]{Yu Chen}
\author[b,\thanks{Corresponding author. E-mail address: xushx3@mail.sysu.edu.cn}] {Shuai-Xia Xu} 
\author[a]{Yu-Qiu Zhao}
\affil[a]{ Department of Mathematics, Sun Yat-sen University, Guangzhou, 510275,
China}
\affil[b]{Institut Franco-Chinois de l'Energie Nucl\'{e}aire, Sun Yat-sen University, Guangzhou, 510275, China}
\date{}
\maketitle

\begin{abstract}
We consider the asymptotics of the  partition function of the  extended Gross-Witten-Wadia unitary matrix model by introducing
an extra logarithmic term in the potential. The partition function can be written as  a Toeplitz determinant with entries expressed in terms of the modified Bessel functions of the first kind and furnishes  a $\tau$-function sequence of the Painlev\'e III$' $ equation. We derive the  asymptotic expansions
of the Toeplitz determinant up to and including the constant terms as the size of the determinant tends to infinity.  The constant terms therein are expressed in terms of the Riemann  zeta-function and the Barnes  $G$-function.  A
 third-order phase transition in the leading terms of the asymptotic expansions is also observed.

\end{abstract}

\tableofcontents

\section{Introduction}

In this paper, we consider the  partition function of the extended Gross-Witten-Wadia unitary matrix model
\begin{equation}\label{eq:URMT}
Z_{n,\nu}(t)=\frac{1}{n!}\left(\prod_{k=1}^{n}\int_{\Gamma}\frac{ds_{k}}{2\pi is_{k}}\right)\Delta(\mathbf{s})\Delta(\mathbf{s}^{-1})e^{\sum\limits_{k=1}^{n} \frac{t}{2}(s_{k}+s_{k}^{-1})+\nu \log s_k},
\end{equation}
where $\nu\in\mathbb{C}$, $t>0$,  $\Delta(\mathbf{s})=\prod_{i<j}(s_{i}-s_{j})$ is the Vandermonde determinant and  the branch of each $\log s_k$ is chosen such that $\arg s_k\in(-\pi,\pi)$.  Here, the path of integration $\Gamma$ is the Hankel  loop, which starts at $-\infty$, encircles the origin once in the positive direction and returns to  $-\infty$; as illustrated in Figure \ref{hankelloop}.
\begin{figure}[h]
  \centering
  \includegraphics[width=6.5cm,height=4cm]{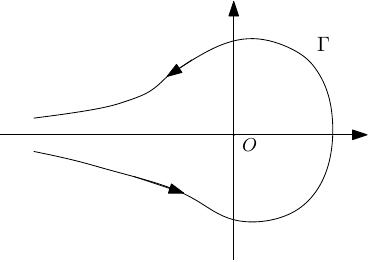}\\
  \caption{The contour $\Gamma$}\label{hankelloop}
\end{figure}

%

When $\nu\in\mathbb{Z}$, the   integration path can be deformed to the unit circle. Therefore,  setting $\nu=0$ simplifies the model \eqref{eq:URMT} to the classical Gross-Witten-Wadia unitary matrix model \cite{GW,W}.
For general $\nu\in\mathbb{Z}$, the partition function \eqref{eq:URMT}  represents
a certain average over the Gross-Witten-Wadia unitary matrix model. The partition function  \eqref{eq:URMT} can be written as the Toeplitz determinant whose entries are the modified Bessel functions; see \eqref{TopBes} below. Using the modified Bessel function as a seed solution,  it was shown that the partition function \eqref{eq:URMT}  is a $\tau$-function sequence in Okamoto's Hamiltonian formulation of Painlev\'e III$'$ equation; see  \cite[Proposition 2]{FW-2003} and \cite{FW}. 



For non-integer value of $\nu$, the  logarithmic function in the potential in \eqref{eq:URMT}
has a branch point at the origin. In \cite[page 165 ]{FW}, to handle this,  the complex plane is cut along the negative real axis and the contour of integration is deformed from the unit circle  to the Hankel loop as shown in Figure \ref{hankelloop}. 
The contour deformation ensures that the partition function \eqref{eq:URMT} keeps the  Toeplitz determinant  structure   with entries given in terms of modified Bessel functions, as shown in \eqref{TopBes} below.  Therefore,  the  structure of $\tau$-function of the Painlev\'e III$' $ equation are maintained for the partition function \eqref{eq:URMT} with general complex value of $\nu$; see \cite[Proposition 2]{FW-2003} and  also  \cite[Theorem 1.1]{CXZ}.

The matrix model \eqref{eq:URMT} with  general parameter $\nu$ has also arisen recently in the  studies of  irregular conformal block  \cite{IOY,IOY19-2, IOY10}.
As shown in \cite[Eqs.(F.3) and (F.4)]{IOY19-2}, the conformal block can be represented by the partition function of certain matrix model of size $n$.
By taking suitable scaling limits of the parameters in the partition function, they arrive at
 the  extended Gross-Witten-Wadia unitary matrix model of size $n$ by introducing an extra logarithmic term in the potential  \cite[Eqs. (F.29) and (F.31)]{IOY19-2}, which is identified with the irregular conformal block.
Generally, the matrix model considered in \cite{IOY19-2} depends on two types of integration loops.  As shown in \cite[Eq. (F.29)]{IOY19-2},
the first $N_1$ out of the $N$ integration contours in the matrix model are along the same loop $f(\Gamma)$ with the transformation $f(z)=-\frac{1}{z}$,  and the remaining  integration contours are along the loop $e^{i\pi}\Gamma$, where $\Gamma$ is the Hankel loop.
If the coefficient of the logarithmic term is an integer, the branch cut vanishes and both the contours can be deformed to the unit circle. The model  then represents
a certain average over the Gross-Witten-Wadia unitary matrix model as mentioned above; see also  \cite[Eqs. (2.6) and (3.1)]{IOY19-2}.
For general parameter $\nu$ and $N_1=0$,  the multiple integral  \cite[Eqs. (F.31)]{IOY19-2} is along the same contour $e^{i\pi}\Gamma$  with
 $\Gamma$  being the Hankel loop, and   the matrix model \cite[Eqs. (F.29) and (F.31)]{IOY19-2} is  equivalent to \eqref{eq:URMT} after a change of variables $z_k=e^{i\pi} s_k$ for $k=1,2,\dots, n$.

Motivated by the applications in the irregular conformal block and  the $\tau$-function theory of Painlev\'e equations, we consider the extended Gross-Witten-Wadia matrix model \eqref{eq:URMT} posed on the Hankel loop. As mentioned before, the contour deformation from the unit circle to the Hankel loop also allows the partition function \eqref{eq:URMT}  to be expressed in the structure of  Toeplitz determinant  whose entries are  the modified Bessel functions with order depending on the parameter $\nu$. To be more precise,
we denote by
 $D_{n,\nu}(t)$ the Toeplitz determinant associated with the weight function 
\begin{equation}\label{eq: weight}
w(z)=e^{\frac{t}{2}(z+\frac{1}{z})+\nu \log z}, \quad  z\in \Gamma, \quad t>0, \quad \nu\in \mathbb{C}, 
\end{equation}
where the branch of $\log z$ is chosen such that $\arg z\in(-\pi,\pi)$. That is,
\begin{equation}\label{eq: Top}
D_{n,\nu}(t)=\det\left(m_{j-i}\right)_{i,j=0}^{n-1},\end{equation}
with the moments
\begin{equation}\label{eq: moments}
m_k=m_k(t)=\int_{\Gamma} s^{k} w(s) \frac{ds}{2\pi is}, \quad k\in\mathbb{Z}.
 \end{equation}
Then,  the partition function \eqref{eq:URMT} can be expressed equivalently in terms of the Toeplitz determinant
\begin{equation}\label{eq: =Top}
Z_{n,\nu}(t)=D_{n,\nu}(t). \end{equation}
From the integral representation of the modified Bessel function of the first kind, also termed the $I$-Bessel function, we obtain
\begin{equation}\label{eq: Bes}
m_k(t)=I_{-k-\nu}(t),  \end{equation}
where $I_{\alpha}(t)=\left(\frac{t}{2}\right)^{\alpha}\sum_{j=0}^{\infty} \frac{  \left(t^2/4\right)^j }  {j!\Gamma(j+1-\alpha)}$ denotes the $I$-Bessel function of order $\alpha$; see \cite[Eqs.(10.9.19), (10.25.2) and (10.27.6)]{Olver}.
Therefore,  we may express the partition function and the  Toeplitz determinant in terms of the determinant of the $I$-Bessel functions
\begin{equation}\label{TopBes}
Z_{n,\nu}(t)= D_{n,\nu}(t)=\det \left(I_{i-j-\nu}(t)\right)_{i,j=0}^{n-1}.
\end{equation}
It was shown  in \cite[Proposition 2]{FW-2003} and   \cite[Theorem 1.1]{CXZ} that  $Z_{n,\nu}(t)$ is a $\tau$-function sequence of the Painlev\'e III$' $ equation with  general parameter $\nu$.


Denote  $\{\pi_n\}_{n\in\mathbb{N}}$ and $ \{\tilde{\pi}_n\}_{n\in\mathbb{N}}$ the families   of monic orthogonal polynomials  defined by the orthogonality on the Hankel loop depicted in Figure \ref{hankelloop}:
 \begin{equation}\label{eq:MOrtho}
\int_{\Gamma} \pi_n(s) \tilde{\pi}_m(s^{-1}) w(s) \frac{ds}{2\pi is}=h_n\delta_{n,m}, \end{equation}
where $w(z)$ is given in \eqref{eq: weight}. 
Suppose the determinant $D_{n,\nu}(t)$ 
does not vanish,
then the orthogonal polynomials can  be constructed explicitly as follows:
\begin{equation}\label{eq: OPexpre}
\pi_n(z)=\frac{1}{D_{n,\nu} (t)}\left| \begin{array}{cccc}m_0 & m_{1} & \cdots & m_{n} \\m_{-1} & m_0 & \cdots & m_{-1+n} \\ \vdots &  \vdots &  \vdots &  \vdots \\
m_{-n+1} & m_{-n+2} & \cdots & m_{1} \\1 & z & \cdots & z^{n}\end{array} \right|,
 \end{equation}
 and
 \begin{equation}\label{eq: tildeOPexpre}
 \tilde{\pi}_n(z)=\frac{1}{D_{n,\nu}(t)}\left| \begin{array}{cccc}m_0 & m_{-1} & \cdots & m_{-n} \\m_1 & m_0 & \cdots & m_{1-n} \\ \vdots &  \vdots &  \vdots &  \vdots \\
m_{n-1} & m_{n-2} & \cdots & m_{-1} \\1 & z& \cdots & z^n\end{array} \right|.
\end{equation}
Then, the Toeplitz determinant can be written in terms of the constants $h_{k}$'s in \eqref{eq:MOrtho}
\begin{equation}\label{TOP-h}
D_{n,\nu}(t)=\prod\limits_{k=0}^{n-1}h_{k}.
\end{equation}
Similar to the orthogonal polynomials on the unit circle \cite[Theorem 11.4.2]{S} and  \cite[Lemma 2.3]{DIK},  we obtain from the 
orthogonality  \eqref{eq:MOrtho} that the orthogonal polynomials $\pi_{n}(z)$ and $\widetilde{\pi}_{n}(z)$ satisfy the following recurrence relations
\begin{equation}\label{three-term1}
\pi_{n+1}(z)=z\pi_{n}(z)+\pi_{n+1}(0)\,\widetilde{\pi}^{*}_n(z),
\end{equation}
\begin{equation}\label{three-term2}
\widetilde{\pi}^{*}_{n+1}(z)=\widetilde{\pi}^{*}_n(z)+\widetilde{\pi}_{n+1}(0)z\pi_{n}(z),
\end{equation}
where  $\widetilde{\pi}^*_{n}(z)=z^{n}\widetilde{\pi}_{n}(z^{-1})$ is the reversed polynomial associated with $\widetilde{\pi}_{n}(z)$.
We also have the Christoffel-Darboux formula 
\begin{equation}\label{eq:CD}
(1-a^{-1}z)\sum_{k=0}^{n-1}p_{k}(z) \tilde{p}_{k}(a^{-1})=a^{-n}p_{n}(a) z^n\tilde{p}_{n}(z^{-1})-\tilde{p}_{n}(a^{-1})p_{n}(z),
\end{equation}
where $p_{n}=\gamma_{n}\pi_{n}$, $\widetilde{p}_{n}=\gamma_{n}\widetilde{\pi}_{n}$,  $\gamma_{n}=h_{n}^{-1/2}$ are the orthonormal 
polynomials associated with the weight function \eqref{eq: weight}. The  Christoffel-Darboux formula can be derived by using  
the recurrence relations \eqref{three-term1} and \eqref{three-term2} in the same way as the orthogonal polynomials on the unit circle \cite[Lemma 2.3]{DIK}.

It is known  that the large-$n$ asymptotic expansion of the  partition function $Z_{n,\nu}(t)$ in \eqref{eq:URMT} with $\nu=0$ for Gross-Witten-Wadia model  exhibits a third-order phase transition \cite{GW}. This can be seen from the discontinuity at $\tau=1$ in the third derivative of the  leading term in the following expansion \cite{GW}
\begin{equation}\label{free-energy}
\lim_{n\rightarrow\infty}\frac{1}{n^2}\left(\log Z_{n,0}(n\tau)-   \frac{\tau^2}{4}\right )=\left\{ \begin{array}{ll}
                  0, & \hbox{$0<\tau<1$},\\
                       \tau-\frac{1}{2}\log\tau-\frac{3}{4}- \frac{\tau^2}{4}, & \hbox{$\tau>1$.}
                    \end{array}
\right.
\end{equation}
The above formula was  proved rigorously later in \cite{J}.
It is remarkable that in the double scaling limit as $n\to\infty$ and $\tau\to1$ in certain related speed, the asymptotics of the  partition function $Z_{n,0}(n\tau)$ can be expressed in terms of the Hastings-McLeod solution of the second Painlev\'e equation
\begin{equation}\label{eq:PII}
u''=2u^{3}+xu +\alpha,
\end{equation}
with the parameter $\alpha=0$; see \cite{C,PS}.

In the work \cite{CXZ}, we study the asymptotics of the partition function $Z_{n,\nu}(t)$ \eqref{eq:URMT} with general parameter $\nu\in\mathbb{C}$.
In the double scaling limit as $n\to\infty$ and $\tau\to1$,
we establish an asymptotic  approximation of the logarithmic derivative of the Toeplitz determinant, expressed in terms of  the Hamiltonian associated with the Hastings-McLeod solution of the  second Painlev\'e equation \eqref{eq:PII} with the parameter dependent on $\nu$. We obtain the result by studying the asymptotics of the  Toeplitz determinant \eqref{TopBes}, using the Deift-Zhou steepest descent analysis for the Riemann-Hilbert (RH, for short) problems for the orthogonal polynomials defined by \eqref{eq:MOrtho}.



In recent years, there has been a considerable amount of interest in the study of asymptotics
of Toeplitz determinants, due to their important applications in various branches of applied mathematics and mathematical physics.
In \cite{DIK}, the asymptotics of the Toeplitz determinants associated with
a general family of weight functions on the unit circle with any given fixed Fisher-Hartwig singularities were
derived. In \cite{BCL}, the asymptotics of the Toeplitz determinants  with general varying weight $e^{-nV(z)}$ perturbed by  Fisher-Hartwig singularities are derived, where the equilibrium measure associated with the potential $V(z)$  is assumed to be supported on the whole unit circle.
The  transition asymptotics of the Toeplitz determinants are established in several different situations where the singularities vary in $n$ \cite{CIK,CK2015,CC,XZ}. The reader is also referred to the survey article \cite{DIK2} for the historic background and applications in the Ising models.  As a follow-up of the investigations in \cite{BCL}, one may also consider broader classes of partition functions on the Hankel loop with more general potential than the special one in the extended Gross-Witten-Wadia matrix model  \eqref{eq:URMT}.
In the general case, one would need to consider the equilibrium measure on the Hankel loop and we will leave this problem to a future investigation.


In the present paper, we consider the asymptotics of the partition function $D_{n,\nu}(n\tau)$ in \eqref{TopBes} of the extended Gross-Witten-Wadia matrix model  as the matrix size $n$ tends to infinity when $0<\tau<1$ and $\tau>1$, separately. We derive the
large-$n$ asymptotic expansions
for $D_{n,\nu}(n\tau)$ up to and including the constant terms. Similar to \eqref{free-energy}, we observe a
 third-order phase transition in the leading terms of the asymptotic expansions near $\tau=1$.  Moreover, we evaluate the constant term in the  asymptotic expansions by using a certain differential identity with respect to the parameter $\nu$.

\subsection{Statement of results}
\begin{thm}\label{thm1}
Let $\nu\in \mathbb{C}$, $t=n\tau$ and $0<\tau<1$, we have the asymptotic approximation of the logarithm of the Toeplitz determinant
associated with \eqref{eq: weight}
\begin{equation}\label{Integral:logDn1}
\begin{aligned}
\log D_{n,\nu}(t)=&\frac{n^2\tau^2}{4}+n\nu\left(\log \frac{1+\sqrt{1-\tau^2}}{\tau} -\sqrt{1-\tau^2}\right)-\frac{\nu^2}{2}\log n +\frac{\nu}{2}\log2\pi\\
&-\frac{\nu^2}{4}\log(1-\tau^2)+\log G(1-\nu)+O\left(\frac{1}{n}\right),~~\nu\neq 1,2,3 \cdots,~~ n\rightarrow\infty,
\end{aligned}
\end{equation}
where  the error term is uniform for $\tau$ in any compact subsets of $(0,1)$.
For $\nu=1,2,3\cdots,$ the asymptotic behavior of ~$\log D_{n,\nu}(t)$ can be obtained from \eqref{Integral:logDn1} by
using the symmetry $D_{n,\nu}(t)=D_{n,-\nu}(t)$ for $  \nu\in\mathbb{N}$. 
\end{thm}

\begin{thm}\label{thm2}
Let $\nu\in \mathbb{C}$, $t=n\tau$ and $\tau>1$, we have the asymptotic approximation of the logarithm of the Toeplitz determinant
associated with \eqref{eq: weight}
\begin{equation}\label{Integral:logDn3}
\begin{aligned}
\log D_{n,\nu}(t)=&n^{2}\left(\tau-\frac{3}{4}-\frac{1}{2}\log \tau\right)-\frac{1}{12}\log n\\
&+\left(\frac{\nu^2}{2}-\frac{1}{8}\right)\log(1-\tau^{-1})
+\zeta'(-1)+O\left(\frac{1}{n}\right ),~~n\rightarrow\infty,
\end{aligned}
\end{equation}
where $\zeta'(t)$ is the    derivative of the Riemann $\zeta$-function and  the error term is uniform for $\tau$ in any compact subsets of $(1,+\infty)$.
\end{thm}

\begin{rem} Similar to \eqref{free-energy}, we observe a
 third-order phase transition in the leading terms of the asymptotic expansions of \eqref{Integral:logDn1}-\eqref{Integral:logDn3}
 \begin{equation}\label{eq:phaseTrans}
\lim_{n\rightarrow\infty}\frac{1}{n^2}\left(\log Z_{n,v}(n\tau)- \frac{\tau^2}{4}\right)=\left\{ \begin{array}{ll}
                   0, & \hbox{$0<\tau<1$},\\
                     \tau-\frac{1}{2}\log\tau-\frac{3}{4}-\frac{\tau^2}{4}, & \hbox{$\tau>1$,}
                    \end{array}
\right.
\end{equation}
where
 \begin{equation}\label{eq:ph1}
\tau-\frac{1}{2}\log \tau-\frac{3}{4}-\frac{1}{4}\tau^2=-\frac{1}{6}(\tau-1)^3+O\left((\tau-1)^4\right), \quad \tau\to1.\end{equation}
In a separate paper \cite{CXZ}, we have shown that the phase transition can be described by  the Hamiltonian associated with the Hastings-McLeod solution of the second Painlev\'e equation \eqref{eq:PII} with the parameter dependent on $\nu$ in the regime where $n\to\infty$ and $\tau\to1$ in a way such that $n^{2/3}(\tau-1)\to s\in\mathbb{R}$.

\end{rem}
The rest of the paper is arranged as follows. In Section \ref{sec:RHforY}, we provide a RH problem $Y(z;  n, t)$ for the
orthogonal polynomials associated with \eqref{eq: weight} and derive the differential identity for the Toeplitz determinant generated by \eqref{eq: weight} with respect to the parameter $\nu$. In Sections \ref{sec:Asycase1} and \ref{sec:Asycase2}, we perform the Deift-Zhou nonlinear steepest descent analysis   \cite{DZ} of the RH problem
for $Y(z; n, n\tau)$ as $n\to\infty$ for the case $0<\tau<1$ and $\tau>1$ separately.  Then, by using the  differential identity and the results of  the nonlinear steepest descent analysis, we  prove Theorems \ref{thm1} and \ref{thm2}
in Sections \ref{sec:proof1} and \ref{sec:proof2}, respectively.  For the convenience of the reader, we collect the Airy and
Parabolic cylinder parametrices in the Appendix.

\section{Riemann-Hilbert problem for the orthogonal polynomials }\label{sec:RHforY}
\begin{rhp} \label{RHP: Y}
We look for a $2 \times 2$ matrix-valued function $Y(z; n, t)$ ($Y(z)$, for short) satisfying the properties.
\begin{itemize}
\item[\rm (1)] $Y(z)$ is analytic in $\mathbb{C} \setminus \Gamma$; where the contour $ \Gamma$ is shown in Figure \ref{hankelloop}.
\item[\rm (2)] $Y(z)$   has continuous boundary values $Y_{+}(z) $ and $Y_{-}(z) $ as $z$ approaches  the contour $ \Gamma$ from the positive and negative sides, respectively.   And they satisfy the jump condition
  \begin{equation}\label{eq:Yjump}
  Y_{+}(z)=Y_{-}(z)
  \begin{pmatrix}
  1 & \frac{w(z )}{z^n} \\
  0 & 1
  \end{pmatrix},\qquad z\in \Gamma,
  \end{equation}
where the weight  function is defined in \eqref{eq: weight}. 
\item[\rm (3)] As $z\to \infty$, we have
  \begin{equation}\label{eq:YInfinity}
  Y(z)=\left (I+\frac{Y_{-1}}{z}+O\left (\frac 1 {z^2}\right )\right)
 \begin{pmatrix}
 z^n & 0 \\
 0 & z^{-n}
 \end{pmatrix}.
 \end{equation}
\end{itemize}
\end{rhp}

It was discovered  by Fokas, Its and Kitaev  \cite{FIK} that the  orthogonal polynomials on the real line can be  represented as a solution of a Riemann-Hilbert problem.  Such a formulation has been extended to the orthogonal polynomials on the circle in \cite{BDJ, DIK}. 
Similarly,  
the unique solution to the above RH problem with jump on the Hankel loop can be constructed as follows \cite{CXZ}
\begin{equation}\label{eq:YSolution}
Y(z)=
\begin{pmatrix}
\pi_n(z) & \frac{1}{2\pi i } \int_{\Gamma}\frac{\pi_n(x) w(x)dx}{x^n(x-z)}\\
-h_{n-1}^{-1} \widetilde{\pi}^*_{n-1}(z)&  -\frac{h_{n-1}^{-1}}{2\pi i}  \int_{\Gamma}\frac{\widetilde{\pi}^*_{n-1}(x) w(x)dx}{x^n(x-z)}
\end{pmatrix},
\end{equation}
where $\widetilde{\pi}^*_{n-1}(z)=z^{n-1}\widetilde{\pi}_{n-1}(z^{-1})$ and $\pi_n$, $\widetilde{\pi}_{n-1}$  and $h_{n-1}$  are defined by \eqref{eq:MOrtho}.
The  orthogonal polynomials $\pi_n$ and $\widetilde{\pi}_{n-1}$ exist if $D_{n-1,\nu}(t)\neq0$ and $D_{n,\nu}(t)\neq0$ as seen from \eqref{eq: OPexpre} and \eqref{eq: tildeOPexpre}. For positive weight function on the unit circle, the Toeplitz
determinants  are strictly positive and the orthogonal polynomials exist. 
In our case, we will show later that the  RH problem for $Y(z)$ can be solved asymptotically for $n$ large enough and $t$ in the regions described in  Theorems \ref{thm1} and \ref{thm2}, which justifies the existence of the orthogonal polynomials on the Hankel loop for large enough $n$.


We now derive a differential identity  for the logarithmic derivative of the Toeplitz
determinant with respect to the parameter $\nu$ inspired  by the earlier works \cite{DIK2014, IK, K}.
\begin{pro} \label{Pro: DifferentialIdentity} Let $n\in \mathbb{N}$ be fixed and assume that $D_{k,\nu}(t)\neq0$, for $ k=n-1,n,n+1$, 
then we have the following differential identity with respect to the parameter $\nu$
\begin{align}\label{eq:dnuD}
\frac{d}{d\nu}\log D_{n,\nu}(t)=&n\frac{d}{d\nu}\log Y_{12}(0;n)-\frac{t}{2}\frac{d}{d\nu}\left((Y_{-1})_{11}+\frac{Y_{21}'(0;n+1)}{Y_{21}(0;n+1)}\right)\\ \nonumber
&-\frac{t}{2}\left(Y_{11}(0;n+1)\frac{d}{d\nu}Y_{22}(0; n)+(Y_{22}(0;n+1)\frac{d}{d\nu}Y_{11}(0; n)\right),
\end{align}
where $Y_{-1}$ is coefficient in the expansion \eqref{eq:YInfinity} and  $Y_{21}'$ denotes the derivative of $(2, 1)$-entry of $Y(z)$ with respect to $z$.
\end{pro}
 \begin{proof} We first prove the  differential identity under the condition $D_{k,\nu}(t)\neq0, k=0,1, \dots,n+1$. 
From \eqref{eq:MOrtho}, we see that the polynomials
$p_{n}=\gamma_{n}\pi_{n}$, $\widetilde{p}_{n}=\gamma_{n}\widetilde{\pi}_{n}$,  $\gamma_{n}=h_{n}^{-1/2}$ satisfy
 \begin{equation}\label{eq:Ortho}
\int_{\Gamma} p_n(s) \tilde{p}_m(s^{-1})\frac{ w(s) ds}{2\pi is}=\delta_{n,m},\quad m=0,1,\cdots,n.
\end{equation}
This, together with \eqref{TOP-h}, implies that
\begin{align}\label{eq:DDn}\nonumber
\frac{d}{d\nu}\log D_{n,\nu}&=-2\sum_{k=0}^{n-1} \frac{ 1 }{\gamma_k}\frac{d\gamma_{k}}{d\nu}\\&\nonumber  =-\sum_{k=0}^{n-1} \int_{\Gamma} \frac{d p_{k}(s)}{d\nu} \tilde{p}_k(s^{-1}) w(s) \frac{ds}{2\pi is}-
\sum_{k=0}^{n-1} \int_{\Gamma} p_{k}(s) \frac{d\tilde{p}_{k}(s^{-1})}{d\nu} w(s) \frac{ds}{2\pi is}\\
&=-\int_{\Gamma}\frac{d}{d\nu}\left(\sum_{k=0}^{n-1}p_{k}(s) \tilde{p}_{k}(s^{-1})\right)\frac{w(s) ds}{2\pi is}.
\end{align}
By the Christoffel-Darboux formula \eqref{eq:CD}, we have 
\begin{equation}\label{Ch-D}
\sum_{k=0}^{n-1}p_{k}(s) \tilde{p}_{k}(s^{-1})=-np_{n}(s) \tilde{p}_{n}(s^{-1})+s\left(\tilde{p}_{n}(s^{-1})\frac{d }{ds}p_{n}(s)-p_{n}(s) \frac{d}{ds}\tilde{p}_{n}(s^{-1})\right).
\end{equation}
Substituting  \eqref{Ch-D} into \eqref{eq:DDn}, and using the orthogonality \eqref{eq:MOrtho}, we obtain
\begin{equation}\label{eq:DnuDn}
\frac{d}{d\nu}\log D_{n,\nu}=-\int_{\Gamma} s\frac{d p_{n}(s)}{ds} \frac{d\tilde{p}_{n}(s^{-1})}{d\nu} \frac{ w(s)ds}{2\pi is}+\int_{\Gamma}\frac{d p_{n}(s)}{d\nu}\left\{ s\frac{d\tilde{p}_{n}(s^{-1})}{ds}\right\}\frac{ w(s) ds}{2\pi is}.
\end{equation}
To simplify the first integral on the right-hand side of the above equation, we derive by using integration by parts
\begin{align}\label{I1}\nonumber
\int_{\Gamma}& s\frac{d p_{n}(s)}{ds} \frac{d\tilde{p}_{n}(s^{-1})}{d\nu} \frac{ w(s)ds}{2\pi is}\\
=&-\int_{\Gamma}  p_{n}(s) \left\{s\frac{d^2}{dsd\nu} \tilde{p}_{n}(s^{-1}) \right\}\frac{ w(s) ds}{2\pi is}-\nu\int_{\Gamma}  p_{n}(s) \frac{d\tilde{p}_{n}(s^{-1})}{d\nu} \frac{ w(s)ds}{2\pi i s}\\ \nonumber
&+ \frac{t}{2}\int_{\Gamma}p_{n}(s) \left\{s^{-1}\frac{d\tilde{p}_{n}(s^{-1})}{d\nu}\right\}\frac{w(s) ds}{2\pi is}
- \frac{t}{2}\int_{\Gamma}sp_{n}(s) \frac{d\tilde{p}_{n}(s^{-1})}{d\nu} \frac{w(s) ds}{2\pi i s}.\\ \nonumber
\end{align}
We calculate the integrals on the right-hand side of the above equation term by term. Denote
$\tilde{\pi}_{n}(s^{-1})=s^{-n}+\tilde{a}_{n,n-1}s^{-n+1}\cdots$.  We obtain from $\widetilde{p}_{n}=\gamma_{n}\widetilde{\pi}_{n}$ and the  orthogonality  \eqref{eq:Ortho} that
\begin{align}\label{I11}
\int_{\Gamma} p_{n}(s)\left\{ s\frac{d^2}{dsd\nu}\tilde{p}_{n}(s^{-1})\right\}\frac{ w(s) ds}{2\pi is}=-\frac{n}{\gamma_{n}} \frac{d\gamma_{n}}{d\nu},
\end{align}
\begin{align}\label{I12}
\nu\int_{\Gamma}  p_{n}(s) \frac{d\tilde{p}_{n}(s^{-1})}{d\nu}\frac{w(s)ds}{2\pi i s}=\frac{\nu}{\gamma_{n}} \frac{d\gamma_{n}}{d\nu},
\end{align}
and
\begin{align}\label{I13}\nonumber
&\int_{\Gamma}p_{n}(s) \left\{s^{-1}\frac{d\tilde{p}_{n}(s^{-1})}{d\nu}\right\}\frac{w(s) ds}{2\pi is}
=\int_{\Gamma}p_{n}(s) \left\{\frac{d}{d\nu}\gamma_{n} s^{-(n+1)}+\frac{d}{d\nu}(\gamma_{n}\tilde{a}_{n,n-1}) s^{-n}\right\}\frac{w(s) ds}{2\pi is}
\\ \nonumber
&=\int_{\Gamma}p_{n}(s)\left(\frac{d\gamma_{n}}{d\nu}\tilde{\pi}_{n+1}(s^{-1})+\left(\frac{d(\gamma_{n}\tilde{a}_{n,n-1})}{d\nu}-
\tilde{a}_{n+1,n}\frac{d\gamma_{n}}{d\nu} \right)\tilde{\pi}_{n}(s^{-1})\right)\frac{w(s) ds}{2\pi is}\\
&=\frac{1}{\gamma_{n}}\left(\frac{d(\gamma_{n}\tilde{a}_{n,n-1})}{d\nu}-
\tilde{a}_{n+1,n}\frac{d\gamma_{n}}{d\nu}\right ).
\end{align}
To simplify  the last integral in \eqref{I1}, we obtain by using  the orthogonality \eqref{eq:Ortho} and the recurrence relation \eqref{three-term1} 
\begin{align}\label{I14}\nonumber
\int_{\Gamma}sp_{n}(s) \frac{d\tilde{p}_{n}(s^{-1})}{d\nu} \frac{w(s) ds}{2\pi is} &=\int_{\Gamma}(\gamma_{n}\pi_{n+1}(s)-\pi_{n+1}(0)s^{n}\tilde{p}_{n}(s^{-1}))\frac{d\tilde{p}_{n}(s^{-1})}{d\nu} \frac{w(s) ds}{2\pi i s}\\ 
&=-\frac{\pi_{n+1}(0)}{{\gamma_{n}}} \frac{d(\gamma_{n}\tilde{\pi}_{n}(0))}{d\nu}.
\end{align}
Substituting the above expressions into \eqref{I1}, we obtain
\begin{align}\label{I1234}\nonumber
\int_{\Gamma}& s\frac{d p_{n}(s)}{ds} \frac{d\tilde{p}_{n}(s^{-1})}{d\nu} \frac{ w(s)ds}{2\pi is}\\
&=\frac{n-\nu}{\gamma_{n}} \frac{d\gamma_{n}}{d\nu} +\frac{t\pi_{n+1}(0)}{2{\gamma_{n}}} \frac{d(\gamma_{n}\tilde{\pi}_{n}(0))}{d\nu}
+\frac{t}{2\gamma_{n}}
\left(\frac{d(\gamma_{n}\tilde{a}_{n,n-1})}{d\nu}-
\tilde{a}_{n+1,n}\frac{d\gamma_{n}}{d\nu}\right ).
\end{align}
A similar calculation yields
\begin{align}\label{I2}\nonumber
\int_{\Gamma}&\frac{d p_{n}(s)}{d\nu}\left\{s\frac{d\tilde{p}_{n}(s^{-1})}{ds}\right\} \frac{ w(s)ds}{2\pi is}\\
=&\frac {-n-\nu}{\gamma_{n}} \frac{d\gamma_{n}}{d\nu} -\frac{t\tilde{\pi}_{n+1}(0)}{2\gamma_{n}} \frac{d(\gamma_{n}\pi_{n}(0))}{d\nu}
-\frac{t}{2\gamma_{n}}
\left(\frac{d(\gamma_{n}a_{n,n-1})}{d\nu}-a_{n+1,n}\frac{d\gamma_{n}}{d\nu}\right),
\end{align}
where $a_{n,n-1}$  and $\tilde{a}_{n,n-1}$ denote the sub-leading coefficients of $\pi_n(z)$ and $\tilde{\pi}_{n}(z)$, respectively.
From  the recurrence relations \eqref{three-term1} and \eqref{three-term2}, we have
\begin{equation}\label{eq:an}
a_{n+1,n}-a_{n,n-1}-\pi_{n+1}(0)\tilde{\pi}_{n}(0)=0, \quad \tilde{a}_{n+1,n}-\tilde{a}_{n,n-1}-\tilde{\pi}_{n+1}(0)\pi_{n}(0)=0.
\end{equation}

Thus, we obtain \eqref{eq:dnuD} under the condition $D_{k,\nu}(t)\neq0, k=0,1, \dots,n+1$, by substituting \eqref{I1234}, \eqref{I2} and \eqref{eq:an} into \eqref{eq:DnuDn} and using  \eqref{eq:YSolution}.  
Substituting into \eqref{TopBes} the expression of the modified Bessel function given after \eqref{eq: Bes}, we have  the asymptotic behavior of $D_{k,\nu}(t)$ as $t\to 0$
$$D_{k,\nu}(t)\sim c_{k,\nu} t^{-k\nu},$$
where $c_{k,\nu}\neq 0$, except possibly a discrete set of $\nu$. Therefore,  if $c_{k,\nu}\neq 0$ the zeros of $D_{k,\nu}(t)$ are isolated since $D_{k,\nu}(t)$  is an analytic function in $t$.  While the existence of $Y(z; n)$  and $Y(z; n+1)$ is guaranteed by the weaker condition  $D_{k,\nu}(t)\neq0, k=n-1, n,n+1$.  Therefore,
we can obtain \eqref{eq:dnuD} under this weaker condition by an argument of analytic continuation. 
This completes the proof of the proposition.
\end{proof}

\section{Asymptotics of $Y(z;n\tau)$: case  $0<\tau<1$}\label{sec:Asycase1}
In this section, we will perform the Deift-Zhou nonlinear steepest descent analysis  \cite{DZ} of
the RH problem for $Y(z;n\tau)$
for $\tau$ in any compact subset of $(0,1)$.

First, we   modify the contour in the RH problem by introducing $\widehat{Y}$ as follows:
\begin{equation}\label{eq:Yhat}
  \widehat{Y}(z)=\left\{ \begin{array}{ll}
                      Y(z), & \quad z\in  \Omega_{0}\cup \Omega_{2}, \\ 
                     Y(z) \begin{pmatrix}1&-\frac{w(z)}{z^{n}}\\0 &1\end{pmatrix}, & \quad z\in\Omega_{1};
                    \end{array} 
\right.
 \end{equation}
see Figure \ref{hankelcontour} for the domains. Then we obtain the following RH problem for $\widehat{Y}$.
\begin{figure}[h]
  \centering
  \includegraphics[width=8cm,height=4.1cm]{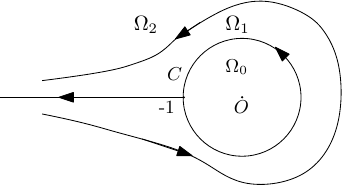}\\
  \caption{The regions for  $\Omega_{k}$, $k=0,1,2$.}\label{hankelcontour}
\end{figure}

\begin{rhp}\label{RHP: Yhat}

\item[\rm (1)] $\widehat{Y}(z)$ is analytic in $\mathbb{C}\setminus \Sigma_{\widehat{Y}} $, where the contours   $\Sigma_{\widehat{Y}}= (-\infty,-1)\cup C$
are depicted in Figure \ref{Yhatcontour}, of which $C$ is the unit circle centered at the origin.
\item[\rm (2)] $\widehat{Y}(z)$  has continuous boundary values on $\Sigma_{\widehat{Y}}$, which satisfy $\widehat{Y}_{+}(z)=\widehat{Y}_{-}(z)J_{\widehat{Y}}(z)$ with
  \begin{equation}\label{eq:Yhatjump}
 J_{\widehat{Y}}(z)=\left\{
\begin{aligned}
&\begin{pmatrix}1 & \frac{w(z)}{z^{n}}\\ 0 & 1 \end{pmatrix},\quad &z&\in C,\\
&\begin{pmatrix}1 & \frac{w_{+}(z)-w_{-}(z)}{z^{n}}\\ 0 & 1 \end{pmatrix},\quad &z&\in (-\infty,-1), 
\end{aligned}
\right.
  \end{equation}
   where $w_{+}(z)-w_{-}(z)=|z|^{\nu}e^{\frac{t}{2}(z+\frac{1}{z})}(e^{\pi i \nu}-e^{-\pi i \nu})$.
  \item[\rm (3)] As $z\to \infty$, we have
  \begin{equation}\label{eq:UInfinity}
  \widehat{Y}(z)=\left (I+\frac{\widehat{Y}_{-1}}{z}+O\left (\frac 1 {z^2}\right )\right )z^{n\sigma_{3}},
 \end{equation}
where $\sigma_{3}$ is one of the Pauli matrices
\begin{equation}\label{Pauli}
\sigma_{1}=\begin{pmatrix}0 & 1\\ 1 & 0\end{pmatrix},~\sigma_{2}=\begin{pmatrix}0 & -i\\ i & 0\end{pmatrix},
~\sigma_{3}=\begin{pmatrix}1 & 0\\ 0 & -1\end{pmatrix}.
\end{equation}
\end{rhp}
\begin{figure}[h]
  \centering
  \includegraphics[width=8cm,height=2.9cm]{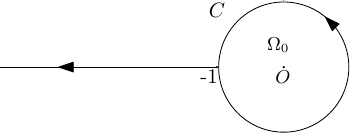}\\
  \caption{Contours of RH problem for $\widehat{Y}$}\label{Yhatcontour}
\end{figure}
\subsection{Normalization: $\widehat{Y} \to T $}
Let $t=n\tau$, $\tau\in(0,1)$, we introduce the first transformation $\widehat{Y}\to T$ to normalize the large-$z$ behavior of $\widehat{Y}(z)$,
\begin{equation}\label{eq:YhatT}
  T(z)=\left\{ \begin{array}{ll}
                      \widehat{Y}(z)e^{\frac{n\tau}{2z}\sigma_3}z^{-n\sigma_3}, & \hbox{$|z|>1$},\\
                     \widehat{Y}(z)e^{\frac{n\tau}{2}z\sigma_3}, & \hbox{$|z|<1$.}
                    \end{array}
\right.
 \end{equation}
 Then, we arrive at the following RH problem for $T(z)$.
\begin{rhp} \label{RHP: T}
The function $T(z)$ defined in \eqref{eq:YhatT} satisfies the following RH problem.
\begin{itemize}
\item[\rm (1)] $T(z)$ is analytic in $\mathbb{C} \setminus \Sigma_{\widehat{Y}}$.
\item[\rm (2)] $T(z)$  has continuous boundary values on $ \Sigma_{\widehat{Y}}$ and they are related by 
  \begin{equation}\label{eq:TJump}
  T_{+}(z)=T_{-}(z)\left\{
  \begin{aligned}
  &\begin{pmatrix}e^{n\phi(z)} & z^{\nu} \\0 & e^{-n\phi(z)}\end{pmatrix},\quad &z&\in C,\\
  &\begin{pmatrix}1 & |z|^{\nu}e^{n\phi(z)}(e^{\pi i \nu}-e^{-\pi i \nu})\\ 0 & 1 \end{pmatrix},\quad &z&\in (-\infty,-1),
  \end{aligned} 
\right.
  \end{equation}
  where
  \begin{equation}\label{eq:phi}
  \phi(z)=\frac{\tau}{2}\left(z-z^{-1}\right)+\log z,
 \end{equation}
and the branch of $\log z$ is chosen such that $\arg z\in(-\pi,\pi)$. For $z\in(-\infty,0)$, we understand $e^{n\phi(z)}$ as 
$e^{n\phi(z)}=e^{n\phi_+(z)}=e^{n\phi_-(z)}$.
\item[\rm (3)] As $z\to \infty$, we have
  \begin{equation}\label{eq:TInfinity}
  T(z)=I+O\left (\frac 1 {z}\right ).
 \end{equation}
 \end{itemize}
\end{rhp}
It is easy to see that $\phi(z)$ has two stationary points on the real axis
\begin{equation}\label{eq:zpm}
z^{\pm}=\frac{-1\pm\sqrt{1-\tau^2}}{\tau},
\end{equation}
where $0<\tau<1$.  For later use, we derive the following properties for the function   $\phi(z)$.
\begin{pro}
The $\phi$-function defined in \eqref{eq:phi} possesses the following properties
\begin{equation}\label{eq:phiSgn}
\left\{ \begin{array}{ll }
\Re  \phi(z)=0,  &|z|=1,  \\ 
\Re  (\phi(z)-\phi(z^{-}))<0,  &|z|=-z^{+},  \\
\Re (\phi(z)-\phi(z^{-}))>0,  &|z|=-z^{-},~ z\not= z^{-},  \\
\Re  (\phi(z)-\phi(z^{-}))<0, &z\in (-\infty,z^{+}).
  \end{array}
\right.
\end{equation}
\end{pro}
\begin{proof}
The first equality in \eqref{eq:phiSgn} follows directly from \eqref{eq:phi}. From \eqref{eq:phi} and \eqref{eq:zpm}, we see that
 $\Re \phi(z)$ has the maximum  at $z=z^{-}$ for $z\in(-\infty, z^{+})$. Therefore, we have the last inequality in \eqref{eq:phiSgn}.
 Moreover, we have $\Re \phi(z^-)>0$ and $\Re \phi(z^+)<0$, since $\Re \phi(-1)=0$ and $z^-<-1<z^+<0$. The other inequalities in \eqref{eq:phiSgn} are verified in a straightforward manner. For example, for $\rho =-z^+=-1/z^-$, substituting   $z=\rho e^{i\vartheta}$ for $\vartheta\in (-\pi, \pi]$ into \eqref{eq:phi} yields
 \begin{equation*}
 \Re  (\phi(z)-\phi(z^{-}))=\frac\tau 2 \left(\frac 1 \rho-\rho\right)(1-\cos\vartheta)+2\log\rho\leq  \tau   \left(\frac 1 \rho-\rho\right) +2\log\rho=2\Re \phi(z^+)<0.
 \end{equation*}This gives  the second inequality. While for $z=\frac 1 \rho e^{i\vartheta}$ with $\vartheta\in (-\pi, \pi]$, it holds
 \begin{equation*}
 \Re  (\phi(z)-\phi(z^{-}))=\frac\tau 2 \left(\frac 1 \rho-\rho\right)(1+\cos\vartheta)>0~~\mbox{for}~\vartheta\not=\pi,
 \end{equation*}
 and the third inequality in  \eqref{eq:phiSgn} follows accordingly.  This completes the proof of the proposition.
\end{proof}
\subsection{Deformation: $T\to S$}
It is seen from  \eqref{eq:TJump}  that the diagonal entries of the jump matrix for $T$  are highly oscillating as $n\to\infty$.
 To transform the oscillating entries  to exponential decay ones on certain contours, we introduce
 the second transformation $T\to S$.
The transformation is based on the following factorization of jump matrix
\begin{equation}\label{eq:Factorization}
  \begin{pmatrix}
  e^{n\phi(z)} & z^{\nu} \\
  0 & e^{-n\phi(z)}
  \end{pmatrix}=   \begin{pmatrix}
  1 &0 \\
  z^{-\nu} e^{-n\phi(z)} & 1
  \end{pmatrix}  \begin{pmatrix}
  0& z^{\nu} \\
  -z^{-\nu} &0
  \end{pmatrix} \begin{pmatrix}
  1 &0 \\
  z^{-\nu} e^{n\phi(z)} & 1
  \end{pmatrix}.
  \end{equation}
  We define the transformation
  \begin{equation}\label{eq:TS}
  S(z)=\left\{ \begin{array}{lll}
  e^{-\frac{n}{2}\phi(z^{-})\sigma_{3}}T(z) \begin{pmatrix}1 & 0 \\z^{-\nu} e^{-n\phi(z)} & 1\end{pmatrix}e^{\frac{n}{2}\phi(z^{-})\sigma_{3}} , & \hbox{for $z\in  \Omega_E$,}\\[.5cm]
  e^{-\frac{n}{2}\phi(z^{-})\sigma_{3}}T(z)\begin{pmatrix}1 & 0 \\- z^{-\nu} e^{n\phi(z)} &1\end{pmatrix} e^{-\frac{n}{2}\phi(z^{-})\sigma_{3}},   &   \hbox{for $z\in  \Omega_I$,}\\[.4cm]
  e^{-\frac{n}{2}\phi(z^{-})\sigma_{3}}T(z)e^{-\frac{n}{2}\phi(z^{-})\sigma_{3}}, &   \hbox{for $z\in  \Omega_0$,}\\[.3cm]
  e^{-\frac{n}{2}\phi(z^{-})\sigma_{3}}T(z)e^{\frac{n}{2}\phi(z^{-})\sigma_{3}}, & \hbox {for $z\in \mathbb{C}\setminus(\overline{\Omega}_E\cup\overline{\Omega}_I\cup\overline{\Omega}_0\cup(-\infty,z^-))$,} \\
  \end{array}
\right.
 \end{equation}
 with $z^{-}$ given in \eqref{eq:zpm} and the regions shown in Figure \ref{Scontour}. 
 As mentioned after \eqref{eq:phi},  we understand $e^{n\phi(z)}$ as 
$e^{n\phi(z)}=e^{n\phi_+(z)}=e^{n\phi_-(z)}$  for $z\in(-\infty,0)$.
Then, $S(z) $ solves
 the following RH problem.
 \begin{figure}[ht]
  \centering
  \includegraphics[width=8cm,height=5cm]{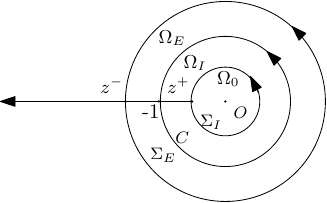}\\
  \caption{Contours and regions of the RH problem for $S(z)$}\label{Scontour}
\end{figure}
\begin{rhp} \label{RHP: S}
The function $S(z)$ defined in \eqref{eq:TS} satisfies the following properties.
\begin{itemize}
\item[\rm (1)] $S(z)$ is analytic in $\mathbb{C} \setminus \Sigma_S$, where $\Sigma_S=\Sigma_E\cup C \cup \Sigma_I\cup(-\infty,z_{+})$ with $\Sigma_E=\{z\in \mathbb{C}:|z|=|z^{-}|\}$ and $\Sigma_I=\{z\in \mathbb{C}:|z|=|z^{+}|\}$  being the boundary of the lens-shaped regions  $\Omega_E$ and $\Omega_I$ as indicated in Figure \ref{Scontour}.
\item[\rm (2)] $S(z)$ has continuous boundary values on $\Sigma_S$ and they are related by
  \begin{equation}\label{eq:SJump}
  S_{+}(z)=S_{-}(z)
  \left\{ \begin{array}{lll}
                    \begin{pmatrix}
  1 & 0 \\
 z^{-\nu} e^{-n(\phi(z)-\phi(z^{-}))} & 1
  \end{pmatrix} , & \hbox{for $z\in  \Sigma_E$,}\\[.5cm]
   \begin{pmatrix}
  0& z^{\nu} \\
  -z^{-\nu} &0
  \end{pmatrix} , & \hbox{for $z\in C$,} \\[.5cm]
 \begin{pmatrix}
  1 & 0 \\
   z^{-\nu} e^{n(\phi(z)-\phi(z^{-}))} &1
  \end{pmatrix}  ,   &   \hbox{for $z\in  \Sigma_I$,}\\[.5cm]
  \begin{pmatrix}1 & |z|^{\nu}e^{n(\phi(z)-\phi(z^{-}))}(e^{\pi i \nu}-e^{-\pi i \nu})\\ 0 & 1 \end{pmatrix},& \hbox{for $z\in  (-\infty,z^{-})$},\\
  \begin{pmatrix}e^{2\pi i\nu} & |z|^{\nu}e^{n(\phi(z)-\phi(z^{-}))}(e^{\pi i \nu}-e^{-\pi i \nu})\\ 0 & e^{-2\pi i\nu} \end{pmatrix},& \hbox{for $z\in  (z^{-},-1)$},\\
  \begin{pmatrix}1 & 0\\ |z|^{-\nu}e^{n(\phi(z)-\phi(z^{-}))}(e^{-\pi i \nu}-e^{\pi i \nu}) & 1 \end{pmatrix},& \hbox{for $z\in  (-1,z^{+})$},\\
  \end{array}
\right.
  \end{equation}
  with $\phi(z)$ defined in \eqref{eq:phi}.
\item[\rm (3)] As $z\to \infty$, we have
  \begin{equation}\label{eq:SInfinity}
  S(z)=I+O\left (\frac 1 {z}\right ).
 \end{equation}
\end{itemize}
\end{rhp}

\subsection{Global parametrix: $N$}
By the properties of $\phi(z)$ stated in \eqref{eq:phiSgn}, we see from  \eqref{eq:SJump} that for $z$ on $\Sigma_{I}\cup \Sigma_{E}$ and bounded away from the interval $(z^{-},-1)$,
the jump matrix for $S(z)$ equals to the identity matrix up to an exponentially small term. Neglecting the exponentially small terms, we arrive at an approximating RH problem for the global parametrix $N(z)$.
\begin{rhp} \label{RHP: N}
We look for a $2 \times 2$ matrix-valued function $N(z)$ satisfying the following properties.
\begin{itemize}
\item[\rm (1)] $N(z)$ is analytic in $\mathbb{C} \setminus \left(C\cup(z^{-},-1)\right)$, where $C$ is the unit circle oriented counterclockwise.
\item[\rm (2)] $N(z)$ has continuous boundary values on  $ C\cup(z^{-},-1)$ and they are related by 
  \begin{equation}\label{eq:N1jump}
  N_{+}(z)=N_{-}(z)\left\{ \begin{array}{ll}
 \begin{pmatrix}0& z^{\nu} \\-z^{-\nu} &0\end{pmatrix} ,& z\in C,\\[.4cm]
 \begin{pmatrix}e^{2\pi i\nu}& 0 \\0 &e^{-2\pi i\nu}\end{pmatrix} ,& z\in (z^{-},-1).\\
   \end{array}
\right.
  \end{equation}
  \item[\rm (3)] As $z\to \infty$, we have
  \begin{equation}\label{eq:NInfinity}
  N(z)=I+O\left (\frac 1 {z}\right ).
 \end{equation}
\end{itemize}
\end{rhp}
The solution to the  RH problem for $N(z)$ can be constructed by using elementary functions as follows:
\begin{equation}\label{eq:NSolution}
  N(z)=\left\{ \begin{array}{ll}
                      \left(\frac{z-z^{-}}{z}\right)^{\nu\sigma_3}, & \hbox{$|z|>1$,} \\
                     (z-z^{-})^{\nu \sigma_3}\begin{pmatrix}
 0 & 1 \\
  -1 & 0
  \end{pmatrix}, & \hbox{$|z|<1$.}
                    \end{array}
\right.
 \end{equation}
Here, the branches of the power functions are chosen such that $\arg z \in(-\pi, \pi)$ and  $\arg(z-z_{-})\in(-\pi, \pi)$.
\subsection{Local parametrix: $P$}
In this subsection, we   proceed    to construct a local parametrix $P(z)$ satisfying the same jump condition as $S(z)$ on the jump contours $\Sigma_S$ (see Figure \ref{Pcontour}) in the neighborhood $U(z^{-},\delta)=\{z\in \mathbb{C}:|z-z^{-}|<\delta\}$ of the saddle points $z^{-}$, and matching with $N(z)$ on the boundary $\partial U(z^{-},\delta)$ with $0<\delta<-1-z^{-}$. Therefore, we formulate  the following RH problem for $P(z)$.
\begin{rhp}\label{RHP:P}
\item[\rm (1)] $P(z)$ is analytic in $U(z^{-},\delta)\setminus \Sigma_S$.
\item[\rm (2)] $P(z)$  satisfies the  jump condition
\begin{equation}\label{eq:tildePJump}
   P_{+}(z)= P_{-}(z)
  \left\{ \begin{array}{lll}
                    \begin{pmatrix}
  1 & 0 \\
 z^{-\nu} e^{-n(\phi(z)-\phi(z_{-}))} & 1
  \end{pmatrix} , & \hbox{for $z\in U(z^{-},\delta)\cap \Sigma_E$,}\\[.5cm]
  \begin{pmatrix}1 & |z|^{\nu}e^{n(\phi(z)-\phi(z^{-}))}(e^{\pi i \nu}-e^{-\pi i \nu})\\ 0 & 1 \end{pmatrix},& \hbox{for $z\in  U(z^{-},\delta)\cap(-\infty,z^{-})$},\\[.5cm]
  \begin{pmatrix}e^{2\pi i \nu} & |z|^{\nu}e^{n(\phi(z)-\phi(z^{-}))}(e^{\pi i \nu}-e^{-\pi i \nu})\\ 0 & e^{-2\pi i \nu} \end{pmatrix},& \hbox{for $z\in  U(z^{-},\delta)\cap(z^{-},-1)$}.\\
  \end{array}
\right.
  \end{equation}
   \item[\rm (3)] On the boundary of $U(z^{-},\delta)$, $P(z)$ satisfies the matching condition
  \begin{equation}\label{eq:Matching}
P(z)N(z)^{-1}= n^{-\frac{ \nu}{2}\sigma_3}\left(I+O(n^{-1/2})\right)n^{\frac{ \nu}{2}\sigma_3}.
 \end{equation}
\end{rhp}
\begin{figure}[ht]
  \centering
  \includegraphics[width=12.8cm,height=5cm]{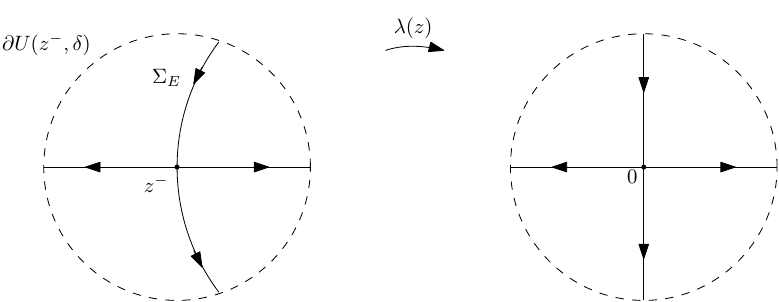}\\
  \caption{The contours and regions of the RH problem for $P(z)$ and their image under the map $\lambda(z)$ }\label{Pcontour}
\end{figure}
To construct the solution of the RH problem for $P(z)$, we introduce the following  conformal mapping
 \begin{equation}\label{eq:lambda}
\lambda(z)=2\left(-\frac{\phi(z)-\phi(z^{-})}{2}\right)^{1/2}\sim (-\phi''(z^{-}))^{\frac{1}{2}}(z-z^{-}), \quad z\to z^{-}.
 \end{equation}
From \eqref{eq:lambda} and \eqref{eq:phi}, we have the expression of $\lambda'(z^{-})$ and $\lambda''(z^{-})$, as $n\rightarrow\infty$,
\begin{equation}\label{asy:lambda1}
\lambda'(z^{-})=(-\phi''(z^{-}))^{\frac{1}{2}}, \quad \lambda''(z^{-})=-\frac{\phi'''(z^{-})}{3(-\phi''(z^{-}))^{\frac{1}{2}}},
\end{equation}
where
\begin{equation}
\phi''(z^{-})=-\frac{\tau}{(z^{-})^{3}}-\frac{1}{(z^{-})^{2}},\quad \phi'''(z^{-})=\frac{3\tau}{(z^{-})^{4}}+\frac{2}{(z^{-})^{3}}.
\end{equation}
Let $\Phi^{(PC)}$ be the parabolic cylinder parametrix given in Appendix \ref{PCP}. Then the parametrix near
$z=z^{-}$ can be constructed as 
\begin{equation}\label{solu:P}
\begin{aligned}
 P(z)=&E(z)(2d_1\nu)^{-\frac{\sigma_3}{2}} \begin{pmatrix}n^{\frac{1}{2}}\lambda(z) & 1 \\ 1 & 0\end{pmatrix} \Phi^{(PC)}(n^{1/2}\lambda(z))\\
 &\times (d_1\nu)^{\frac{\sigma_{3}}{2}}C(z)\sigma_{1}z^{-\frac{\nu}{2}\sigma_{3}}e^{-\frac{n}{2}(\phi(z)-\phi(z^{-}))\sigma_3}e^{\mp\frac{1}{2}\pi i \nu\sigma_{3}}
 \end{aligned}
 \end{equation}
for $\pm \arg z\in(0,\pi),$ where $ d_{1}=\frac{\sqrt{2 \pi}}{\Gamma(-\nu+1)}$ for $\nu \in \mathbb{C} $ and $\nu \neq 1,2, 3\cdots$.
Here the matrix $C(z)=I$ for $\arg z\in(-\frac{\pi}{4},\pi)$ and $C(z)=e^{2\pi i \nu\sigma_{3}}$ for $\arg z\in(\pi,\frac{7\pi}{4})$.
The function $E(z)$ is defined as
\begin{equation}\label{eq:E1}
E(z)=(z-z^{-})^{\nu\sigma_{3}}z^{-\frac{\nu}{2}\sigma_{3}}e^{\pm\frac{1}{2}\pi i \nu\sigma_{3}}\sigma_{1}\lambda(z)^{\nu\sigma_{3}}n^{\frac{1}{2}\nu\sigma_{3}} \end{equation}
for $\pm \arg z\in(0,\pi)$.
Here, the branch for $z^{\nu/2}$ is chosen such that $\arg z\in(-\pi,\pi)$, and the branch of the function $\lambda(z)^{\nu}$ is chosen such that $\arg\lambda(z)\in(-\pi,\pi)$.
It is direct to see that $E(z)$ is analytic in the neighborhood $U(z^{-},\delta)$.
From \eqref{eq:Parab1} and using the recurrence relation \cite[Eq. (12.8.2)]{Olver}
\begin{equation}\label{eq:RecPra}
\frac{z}{2}D_{\nu}(z)+D_{\nu}'(z)=\nu D_{\nu-1}(z),
\end{equation}
we see that
\begin{equation}\label{identity:PC}
 (2\nu)^{-\frac{\sigma_3}{2}} \begin{pmatrix}z& 1 \\ 1 & 0\end{pmatrix} \Phi^{(PC)}(z)\nu^{\frac{\sigma_{3}}{2}} =\begin{pmatrix}
\frac{z}{2}D_{-\nu-1}(i z)+D_{-\nu-1}'(i z) & D_{\nu-1}(z) \\ \nu D_{-\nu-1}(i z) & D_{\nu}(z)\end{pmatrix}
\begin{pmatrix}
e^{i \frac{\pi}{2}(\nu+1)} & 0 \\
0 & 1
\end{pmatrix}
 \end{equation}
 for $\arg z\in(-\frac{\pi}{4},0)$. Therefore, the  parametrix  $P(z)$ is  also well-defined when the parameter $\nu=0$.


From \eqref{solu:P} and \eqref{eq:NSolution}, we have
\begin{equation}\label{P-:match1}
 P(z)N(z)^{-1} =W(z)\begin{pmatrix}1+\frac{\nu(\nu+1)}{2n\lambda(z)^{2}}+O(n^{-2})&
 \frac{1}{n^{\frac{1}{2}}\lambda(z)}+O(n^{-\frac{3}{2}})\\
\frac{\nu}{n^{\frac{1}{2}}\lambda(z)}+O(n^{-\frac{3}{2}})&
1-\frac{\nu(\nu-1)}{2n\lambda(z)^{2}}+O(n^{-2})\end{pmatrix}W(z)^{-1},
\end{equation}
with
\begin{equation}
W(z)=(z-z^{-})^{\nu\sigma_{3}}z^{-\frac{\nu}{2}\sigma_{3}}e^{\pm\frac{1}{2}\pi i \nu\sigma_{3}}\sigma_{1}\lambda(z)^{\nu\sigma_{3}}
n^{\frac{1}{2}\nu\sigma_{3}}d_1^{-\frac{\sigma_{3}}{2}},
\end{equation}
where the error term is uniform for  $z$ on the boundary of $U(z^{-},\delta)$ and for the parameter $\nu$ in any compact subset of $\mathbb{C}\setminus \{1,2,3\cdots\}$.
Therefore, we have
\begin{equation}\label{P-:match}
\begin{aligned}
P(z)N(z)^{-1}
=&n^{-\frac{1}{2}\nu\sigma_{3}}\left(I+\frac 1 {\sqrt n}\begin{pmatrix}0&
\frac{d_1\nu\left(z-z^{-}\right)^{2\nu} z^{-\nu}e^{\pm\pi i\nu}}{\lambda(z)^{2\nu+1}}\\
\frac{\lambda(z)^{2\nu-1} z^{\nu}e^{\mp\pi i\nu}}{d_{1}\left(z-z^{-}\right)^{2\nu}}&0\end{pmatrix}
\right.\\
&\left.+\frac{1}{n}\begin{pmatrix}-\frac{\nu(\nu-1)}{2\lambda^{2}(z)} & 0 \\ 0 &\frac{\nu(\nu+1)}{2\lambda^{2}(z)}\end{pmatrix}
+O\left(n^{-\frac{3}{2}}\right )\right)n^{\frac{1}{2}\nu\sigma_{3}},
\end{aligned}
\end{equation}
where  $ d_{1}=\frac{\sqrt{2 \pi}}{\Gamma(-\nu+1)}$ and $\pm\arg z\in(0,\pi)$.


\subsection{Final transformation: $S\to R$}
We define
\begin{equation}\label{eq:R}
  R(z)=\left\{ \begin{array}{ll}
                     n^{\frac{ \nu}{2}\sigma_3} S(z)N(z)^{-1}n^{-\frac{ \nu}{2}\sigma_3}, & \hbox{$|z-z^{-}|>\delta$,} \\
                   n^{\frac{ \nu}{2}\sigma_3} S(z)P(z)^{-1}n^{-\frac{ \nu}{2}\sigma_3}, & \hbox{$|z-z^{-}|<\delta$.}
                    \end{array}
\right.
 \end{equation}
\begin{figure}[h]
  \centering
  \includegraphics[width=8cm,height=5cm]{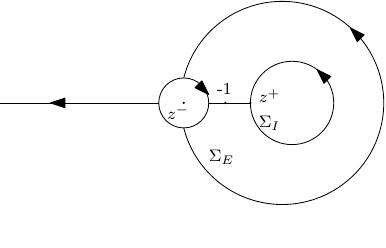}\\
  \caption{ Contours  for the RH problem for $R(z)$}\label{Rcontour}
\end{figure}
Then we have the following RH problem.
\begin{rhp}\label{RHP:R}
\item[\rm(1)] $R(z)$ is analytic for $z\in \mathbb{C}\setminus\Sigma_{R}$, where the contour $\Sigma_{R}$ is illustrated in Figure \ref{Rcontour}.
\item[\rm(2)] On the contour $\Sigma_{R}$, we have
\begin{equation}\label{Jump:R}
R_+(z)=R_-(z)J_{R}(z),
\end{equation}
where
 \begin{equation}\label{eq:JR}
 J_{R}(z)=\left\{\begin{aligned}
&n^{\frac{1}{2}\nu\sigma_{3}}P(z)N(z)^{-1}n^{-\frac{1}{2}\nu\sigma_{3}},\quad &z&\in \partial U(z^{-},\delta),\\
&n^{\frac{1}{2}\nu\sigma_{3}}N(z)J_{S}(z)N(z)^{-1}n^{-\frac{1}{2}\nu\sigma_{3}},\quad &z&\in \Sigma_{R}\setminus\partial U(z^{-},\delta).\\ \end{aligned}
\right.
 \end{equation}
\item[\rm(3)] As $z\rightarrow\infty$, we have
\begin{equation}\label{R:infty}
R(z)=I+\frac{R_1}{z}+O\left(\frac{1}{z^{2}}\right).
\end{equation}
\end{rhp}
From the matching condition \eqref{P-:match} and the properties of $\phi(z)$ stated
in \eqref{eq:phiSgn}, we have the following estimates
\begin{equation}\label{JRestimation}
J_{R}(z)=\left\{\begin{aligned}
&I+O(n^{-\frac{1}{2}}),  &z&\in\partial U(z^{-},\delta),\\
&I+O(e^{-c_1n}),\quad &z&\in \Sigma_{R}\setminus\partial U(z^{-},\delta),
\end{aligned}\right.
\end{equation}
where $c_1$ is a positive constant.
This, together with  \eqref{P-:match} and \eqref{Jump:R}, implies that $J_{R}(z)$ has an expansion of the form for $|z-z^{-}|=\delta$
\begin{equation}\label{JRexpansion}
J_{R}(z)=I+\frac{J^{(1)}_{R}(z)}{n^{ {1}/{2}}}+\frac{J^{(2)}_{R}(z)}{n}+O(n^{- {3}/{2}}), \quad n\rightarrow\infty,
\end{equation}
where
\begin{equation}\label{eq:JR1}
J^{(1)}_{R}(z)=\begin{pmatrix}0 & \frac{d_{1}\nu(z-z^{-})^{2\nu}z^{-\nu}e^{\pm\pi i\nu}}{\lambda(z)^{2\nu+1}} \\
\frac{\lambda(z)^{2\nu-1}z^{\nu}e^{\mp\pi i\nu}}{d_1(z-z^{-})^{2\nu}} & 0\end{pmatrix},
 \end{equation}
for $\pm\arg z\in(0,\pi)$ and
 \begin{equation}\label{eq:JR2}
J^{(2)}_{R}(z)=\begin{pmatrix}-\frac{\nu(\nu-1)}{2\lambda^{2}(z)} & 0 \\ 0 &\frac{\nu(\nu+1)}{2\lambda^{2}(z)}\end{pmatrix}.
 \end{equation}

From \eqref{JRestimation}, we see that the jump matrix $J_{R}(z)$ is close to the identity matrix as $n\to\infty$, with an error term $O(n^{ -{1}/{2}})$ uniformly for $z\in\Sigma_{R}$. Therefore, $R(z)$ satisfies a small-norm RH problem. From the general theory for small-norm RH problems presented in \cite[Section 7.2]{DKMVZ2} and \cite[Section 7.5]{D1}, we see that the RH problem for $R(z)$ is solvable for large enough $n$.
Moreover, using \eqref{JRexpansion}, we have the asymptotic expansion of $R(z)$ in the following form as $n\to\infty$
 \begin{equation}\label{eq:REst}
 R(z)=I+\frac{R^{(1)}(z)}{n^{ {1}/{2}}}+\frac{R^{(2)}(z)}{n}+O(n^{-3/2}),
 \end{equation}
where the error term is  uniform for $z$ in $\mathbb{C}\setminus \Sigma_{R}$ and  the parameter $\nu$ in any compact subset of $\mathbb{C}\setminus \{1,2,3\cdots\}$.

 For later use, we calculate the functions $R^{(1)}(z)$ and $R^{(2)}(z)$ in \eqref{eq:REst}.
 Inserting \eqref{eq:REst} and \eqref{JRexpansion} into the jump condition \eqref{Jump:R} for $R(z)$ yields
%

\begin{equation}\label{Jump:R1}
 R^{(1)}_{+}(z)=R^{(1)}_{-}(z)+J^{(1)}_{R}(z),\quad  R^{(2)}_{+}(z)=R^{(2)}_{-}(z)+R^{(1)}_{-}(z)J^{(1)}_{R}(z)+J^{(2)}_{R}(z),
  \end{equation}
  for $z\in \partial U(z^{-},\delta).$ This, together with the facts that $R^{(1)}(z)=O(z^{-1})$ and $R^{(2)}(z)=O(z^{-1})$ as $z\rightarrow\infty$, implies that
 \begin{equation}\label{Integral:R1}
 R^{(1)}(z)=\frac{1}{2\pi i} \oint_{|x-z^{-}|=\delta}\frac{J^{(1)}_{R}(x)}{x-z}dx,
 \end{equation}
and
  \begin{equation}\label{Integral:R2}
 R^{(2)}(z)=\frac{1}{2\pi i} \oint_{|x-z^{-}|=\delta}\frac{R^{(1)}_{-}(x)J^{(1)}_{R}(x)+J^{(2)}_{R}(x)}{x-z}dx.
 \end{equation}
From the definition of $J^{(1)}_{R}(z)$ given in \eqref{eq:JR1}, we obtain from \eqref{eq:lambda}, \eqref{Integral:R1} and the residue theorem
that
 \begin{equation}\label{solu:R1}
 R^{(1)}(z)=\frac{A^{(1)}}{z-z^{-}},\quad z\in\mathbb{C}\setminus U(z^{-},\delta),
 \end{equation}
 where
 \begin{equation}\label{A1}
 A^{(1)}=\Res\limits_{z=z^{-}}J^{(1)}_{R}(z)=\begin{pmatrix}0 & \frac{d_{1}\nu|z^{-}|^{-\nu}}{\lambda'(z^{-})^{2\nu+1}} \\
d_{1}^{-1}\lambda'(z^{-})^{2\nu-1}|z^{-}|^{\nu}& 0\end{pmatrix}. 
  \end{equation}
  In view of \eqref{JRexpansion} and \eqref{Integral:R2}, we have
 \begin{equation}\label{eq:R2}
  R^{(2)}(z)=\frac{A^{(2)}}{z-z^{-}}+\frac{B^{(2)}}{(z-z^{-})^{2}},\quad z\in\mathbb{C}\setminus U(z^{-},\delta),
  \end{equation}
 where
\begin{equation}\label{A2}
 A^{(2)}=R^{(1)}_{-}(z^{-})A^{(1)}+\Res\limits_{z=z^{-}}J^{(2)}_{R}(z),\quad B^{(2)}=\Res\limits_{z=z^{-}}((z-z^{-})J^{(2)}_{R}(z)).
 \end{equation}
From  \eqref{eq:JR2},  we have
\begin{equation}\label{eq:ResR}
 \Res\limits_{z=z^{-}}J^{(2)}_{R}(z)=-\frac{\lambda''(z^{-})}{\lambda'(z^{-})^{3}}\begin{pmatrix}-\frac{\nu(\nu-1)}{2}&0\\0&\frac{\nu(\nu+1)}{2}\end{pmatrix},
  \end{equation}
  and
\begin{equation}
 \Res\limits_{z=z^{-}}\left((z-z^{-})J^{(2)}_{R}(z)\right)=\frac{1}{\lambda'(z^{-})^{2}}\begin{pmatrix}-\frac{\nu(\nu-1)}{2}&0\\0&\frac{\nu(\nu+1)}{2}\end{pmatrix}.
  \end{equation}
  Therefore, we have
  \begin{equation}\label{eq:B2Exp}
 B^{(2)}=\frac{1}{\lambda'(z^{-})^{2}}\begin{pmatrix}-\frac{\nu(\nu-1)}{2}&0\\0&\frac{\nu(\nu+1)}{2}\end{pmatrix}.
  \end{equation}
  To derive $ A^{(2)}$,  we have from \eqref{Jump:R1} and \eqref{Integral:R1} that
 \begin{equation} \begin{aligned}
 &R^{(1)}_{-}(z^{-})\\
 &=-\lim\limits_{z\rightarrow z^{-}}\left(J^{(1)}_{R}(z)-\frac{A^{(1)}}{z-z^{-}}\right)\\
 &=\begin{pmatrix}0 & \frac{d_{1}\nu|z^{-}|^{-\nu}}{\lambda'(z^{-})^{2\nu+1}}\left(\frac{\nu}{z^{-}}+\frac{(2\nu+1)\lambda''(z^{-})}{2\lambda'(z^{-})}\right)\\ - d_{1}^{-1} |z^{-}|^{\nu}\lambda'(z^{-})^{2\nu-1}\left(\frac{\nu}{z^{-}}+\frac{(2\nu-1)\lambda''(z^{-})}{2\lambda'(z^{-})}\right) & 0\end{pmatrix}. \label{eq:R1z}
 \end{aligned}
  \end{equation}
Substituting \eqref{A1}, \eqref{eq:ResR} and \eqref{eq:R1z}
 into  \eqref{A2}, we have
  \begin{equation}
  \begin{aligned}
  A^{(2)}=&\begin{pmatrix}\frac{\nu}{\lambda'(z^{-})^{2}}(\frac{\nu}{z^{-}}+\frac{(2\nu+1)\lambda''(z^{-})}{2\lambda'(z^{-})}) & 0 \\ 0 & -\frac{\nu}{\lambda'(z^{-})^{2}}(\frac{\nu}{z^{-}}+\frac{(2\nu-1)\lambda''(z^{-})}{2\lambda'(z^{-})})\end{pmatrix}
 \\ &-\frac{\lambda''(z^{-})}{\lambda'(z^{-})^{3}}\begin{pmatrix}-\frac{\nu(\nu-1)}{2}&0\\0&\frac{\nu(\nu+1)}{2}\end{pmatrix}\\
=&\begin{pmatrix}\frac{\nu^{2}}{z^{-}\lambda'(z^{-})^{2}}+\frac{3\nu^2\lambda''(z^{-})}{2\lambda'(z^{-})^{3}}&0\\0&-
\frac{\nu^{2}}{z^{-}\lambda'(z^{-})^{2}}-\frac{3\nu^2\lambda''(z^{-})}{2\lambda'(z^{-})^{3}}\end{pmatrix}.\label{eq:A2Exp}
  \end{aligned}
  \end{equation}

\section{Asymptotics of $Y(z;n\tau)$: case  $\tau>1$}\label{sec:Asycase2}
In this section, we perform the nonlinear steepest descent analysis of the RH problem for  $Y(z;n\tau)$ when $\tau>1$.
The analysis is similar to that performed in $\cite{BDJ}$ where the case $\nu=0$ was considered. 

For later use, we introduce the $g$-function as that used in \cite[Lemma 4.3]{BDJ}:
\begin{equation}\label{g-function}
g(z)=\int_{\xi^{-1}}^{\xi}\log(z-s)\psi(s)ds \quad\text{for} ~z\in \mathbb{C}\setminus \left( (-\infty,-1)\cup \{e^{i\varphi}:-\pi\leq\varphi\leq\theta_{c}\} \right),
\end{equation}
where $\xi=e^{i\theta_{c}}$, $\sin^2\frac{\theta_{c}}{2}=\frac{1}{\tau}, ~0<\theta_{c}<\pi$ and  for each $s=e^{i\theta}$, the branch is chosen such that $\log(z-e^{i\theta})$ is analytic in $\mathbb{C} \setminus \left((-\infty,-1] \cup\{e^{i\varphi}:-\pi\leq\varphi\leq\theta\}\right)$ and $\log(z-e^{i\theta})\sim \log z$  as $z\to+\infty$.
Define $C_1=\{e^{i\varphi}: \theta_c<|\varphi|\leq\pi\}$ and $C_2=\{e^{i\varphi}: -\theta_c\leq\varphi\leq\theta_c\}$. The  density function $\psi(s)$ is given by
\begin{equation}\label{eq:psi}
\psi(s)=\frac{\tau}{4\pi i}\frac{s+1}{s^2}\left(\sqrt{(s-\xi)(s-\xi^{-1})}\right)_{-},\quad s=e^{i\theta}\in C_{2},
\end{equation}
where the branch is chosen such that $\sqrt{(s-\xi)(s-\xi^{-1})}$ is analytic in $\mathbb{C} \setminus C_{2}$
and behaves like $s$ as $s\rightarrow\infty$. Here $\left(\sqrt{(s-\xi)(s-\xi^{-1})}\right)_{-}$  denotes the limits as $s$ approaches $C_2$ from the negative side  with the orientation of $C_2$ specified in Fig. \ref{Yhatcontour1}.
We also define the following  $\phi$-functions
\begin{equation}\label{phi}
\phi(z)=-\frac{\tau}{4}\int_{\xi}^{z}\frac{s+1}{s^2}\sqrt{(s-\xi)(s-\xi^{-1})}ds,\quad z\in \mathbb{C}\setminus\{ C_{2}\cup(-\infty,0)\},
\end{equation}
and
\begin{equation}\label{tildephi}
\widetilde{\phi}(z)=-\frac{\tau}{4}\int_{\xi^{-1}}^{z}\frac{s+1}{s^2}\sqrt{(s-\xi)(s-\xi^{-1})}ds,\quad z\in  \mathbb{C}\setminus\{ C_{2}\cup(-\infty,0)\}.
\end{equation}
Here the paths of integration do not cross the cuts $ C_{2}\cup(-\infty,0)$.
 From the definitions of $\phi(z)$ and $\widetilde{\phi}(z)$, we have the relationship
\begin{equation}
\widetilde{\phi}(z)=\overline{\phi(\overline{z})},\quad z\in \mathbb{C}\setminus C_{2}.
\end{equation}

\begin{figure}[h]
  \centering
  \includegraphics[width=7.5cm,height=3.2cm]{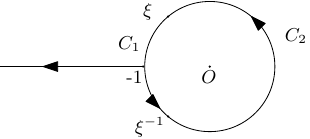}\\
  \caption{The contours $\Sigma_{\widehat{Y}}$ of the RH problem for $\widehat{Y}$}\label{Yhatcontour1}
\end{figure}
According to \cite[Lemma 4.2 and 4.3]{BDJ}, 
we see that $g(z)$ satisfies the Euler-Lagrange equation
\begin{equation}\label{eq:E-L}
g_{+}(z)+g_{-}(z)-V(z)+l=
\left\{
\begin{array}{ll}
\log z+\pi i, & \arg z\in (-\theta_c, \theta_c),\\
\log z-2\phi(z)+\pi i, & \arg z\in (\theta_c, \pi),\\
\log z-2\widetilde{\phi}(z)+\pi i, &\arg z\in ( -\pi,-\theta_c),\\
\log|z|-2\widetilde{\phi}_{-}(z), & z\in (-\infty,-1),  
\end{array}
\right.
\end{equation}
where the Lagrange multiplier $l=-\tau+\log\tau+1$ and the potential $V(z)=-\frac{\tau}{2}(z+\frac{1}{z})$. Moreover, the $g$-function and the
$\phi$-function are related by
\begin{equation}\label{Jump:g}
g_{+}(z)-g_{-}(z)=
\left\{
\begin{array}{ll}
\mp2\phi_{\pm}(z), & \arg z\in (-\theta_c, \theta_c),\\
0,& \arg z\in (\theta_c, \pi),\\
2\pi i, &\arg z\in ( -\pi,-\theta_c),\\
-2\pi i, &z\in (-\infty,-1).
\end{array}
\right.
\end{equation}

\subsection{Normalization: $\widehat{Y} \to T $}
To normalize the large-$z$ behavior of $\widehat{Y}$ in \eqref{eq:Yhat} with the $g$-function defined in \eqref{g-function}, we introduce the transformation $\widehat{Y}\rightarrow T$ as follows
\begin{equation}\label{trans:T}
T(z)=
\left\{
\begin{aligned}
&e^{\frac{nl}{2}\sigma_{3}}\widehat{Y}(z)e^{-ng(z)\sigma_{3}}e^{-\frac{nl}{2}\sigma_{3}},
&\text{for even }~n,\\
&e^{\frac{nl}{2}\sigma_{3}}\sigma_{3}\widehat{Y}(z)\sigma_{3}
e^{-ng(z)\sigma_{3}}e^{-\frac{nl}{2}\sigma_{3}}, &\text{for odd }~n,
\end{aligned}
\right.
\end{equation}
where $l$ is defined in \eqref{eq:E-L}. Then $T(z)$ satisfies the following RH problem.

\begin{rhp}\label{RHP:TT}
\item[\rm(1)] $T(z)$ is analytic for $z\in \mathbb{C}\setminus\Sigma_{\widehat{Y}}$, where the contours $\Sigma_{\widehat{Y}}= (-\infty,-1)\cup C$ 
are depicted in Figure \ref{Yhatcontour1}.
\item[\rm(2)] $T(z)$   has continuous boundary values on  $\Sigma_{\widehat{Y}}$ and they are related by
\begin{equation}\label{Jump:T}
T_{+}(z)=T_{-}(z)\left\{
\begin{aligned}
&\begin{pmatrix}e^{n(g_-(z)-g_+(z))} & (-1)^{n}z^{\nu-n}e^{n(g_+(z)+g_-(z)-V(z)+l)}\\  0& e^{n(g_+(z)-g_-(z))} \end{pmatrix}, &z&\in C,\\
&\begin{pmatrix}e^{n(g_-(z)-g_+(z))} & (-1)^{n}2i\sin(\pi\nu) |z|^{\nu}z^{-n}e^{n(g_+(z)+g_-(z)-V(z)+l)}
\\  0& e^{n(g_+(z)-g_-(z))} \end{pmatrix}, &z&\in(-\infty,-1)
\end{aligned}
\right.
\end{equation}
  where $\nu\in \mathbb{C}$ and $V(z)$ is defined in \eqref{eq:E-L}.
\item[\rm(3)]
$T(z)=I+O(z^{-1}),\quad \mathrm{as}\quad z\rightarrow\infty.$
\end{rhp}
From the properties (\ref{eq:E-L}) and (\ref{Jump:g}), the jump condition in the above RH problem
can be expressed in terms of the function $\phi(z)$ and $\widetilde{\phi}(z)$  as follows:
\begin{equation}
T_{+}(z)=T_{-}(z)\left\{
\begin{aligned}
&\begin{pmatrix}e^{2n\phi_{+}(z)} & z^{\nu}\\  0& e^{2n\phi_{-}(z)} \end{pmatrix},\quad &z&\in C_{2},\\
&\begin{pmatrix}1 & z^{\nu}e^{-2n\phi(z)}\\  0& 1 \end{pmatrix},\quad &z&\in C_{1},\Im z>0,\\
&\begin{pmatrix}1 & z^{\nu}e^{-2n\widetilde{\phi}(z)}\\  0& 1 \end{pmatrix},\quad &z&\in C_{1},\Im z<0,\\
&\begin{pmatrix}1 & |z|^{\nu}(e^{\pi i\nu}-e^{-\pi i \nu})e^{-2n\widetilde{\phi}_{-}(z)}\\  0& 1 \end{pmatrix},\quad &z&\in (-\infty,-1).
\end{aligned}
\right.
\end{equation}
\subsection{Deformation: $T\to S$}
Since $\phi_{\pm}(z)$ are purely imaginary on $C_{2}$, the jump matrix for $T (z)$ on $z\in C_{2}$
possesses highly oscillatory diagonal entries. To remove the oscillation, we introduce the second transformation $T \rightarrow S$, based on a factorization
of the oscillatory jump matrix and a deformation of contours. We define
\begin{equation}\label{eq:TtoS}
  S(z)=\left\{ \begin{array}{ll}
  T(z) \begin{pmatrix}1 & 0 \\z^{-\nu} e^{2n\phi(z)} & 1\end{pmatrix} , & \hbox{for $z\in  \Omega_E$,}\\[.5cm]
  T(z)\begin{pmatrix}1 & 0 \\ -z^{-\nu}e^{2n\phi(z)} &1\end{pmatrix},   &\hbox{for $z\in  \Omega_I$,}\\
  T(z), & \hbox{for $z\in \mathbb{C}\setminus(\Omega_E\cup\Omega_I\cup\Sigma_S)$,} \\
  \end{array}
\right.
 \end{equation}
 where $ \Sigma_S=\Sigma_E\cup \Sigma_I\cup C_{1}\cup C_{2}\cup (-\infty,-1)$, with $\Sigma_E$ and $\Sigma_I$ being the boundaries  of the lens-shaped regions  $\Omega_E$ and $\Omega_I$ as indicated in Figure \ref{tauS2}.
 
 \begin{figure}[ht]
  \centering
  \includegraphics[width=11cm,height=6.2cm]{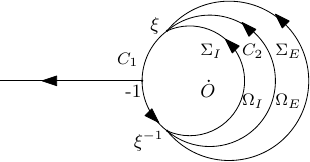}\\
  \caption{Contours and regions for the RH problem for $S(z)$}\label{tauS2}
\end{figure}

Then $S(z)$ solves the following RH problem.
\begin{rhp}\label{RHP:SS}
\item[\rm (1)] $S(z)$ is analytic in $\mathbb{C} \setminus \Sigma_S$; see Figure \ref{tauS2}.
\item[\rm (2)] $S(z)$  has continuous boundary values on  $\Sigma_S$ and they are related by
  \begin{equation}\label{Jump:S}
  S_{+}(z)=S_{-}(z)
  \left\{ \begin{array}{lll}
\begin{pmatrix}1 & 0 \\ z^{-\nu} e^{2n\phi(z)} & 1\end{pmatrix} , & \hbox{for $z\in  \Sigma_E$,}\\[.5cm]
\begin{pmatrix}0& z^{\nu} \\-z^{-\nu} &0\end{pmatrix} , & \hbox{for $z\in C_{2}$,} \\[.5cm]
\begin{pmatrix}1&z^{\nu}e^{-2n\phi(z)}\\0 &1\end{pmatrix} , & \hbox{for $z\in C_{1},~\Im z>0$,} \\[.5cm]
\begin{pmatrix}1&z^{\nu}e^{-2n\widetilde{\phi}(z)}\\0 &1\end{pmatrix} , & \hbox{for $z\in C_{1},~\Im z<0$,} \\[.5cm]
 \begin{pmatrix}1 &0\\ z^{-\nu} e^{2n\phi(z)}  &1\end{pmatrix}  ,   &   \hbox{for $z\in  \Sigma_I$,}\\[.5cm]
 \begin{pmatrix}1 & |z|^{\nu}(e^{\pi i \nu}-e^{-\pi i \nu})e^{-2n\widetilde{\phi}(z)}\\  0& 1 \end{pmatrix},&   \hbox{for $z\in  (-\infty,-1)$,}\\[.5cm]
  \end{array}
\right.
  \end{equation}
  with $\phi(z)$ and $\widetilde{\phi}(z)$ defined in \eqref{phi} and \eqref{tildephi}.
\item[\rm (3)] As $z\to \infty$, we have
  \begin{equation}\label{eq:S2Infinity}
  S(z)=I+O\left (\frac 1 {z}\right ).
 \end{equation}
\end{rhp}
It is seen from \eqref{phi} and \eqref{tildephi} that $\phi(z)$ and $\widetilde{\phi}(z)$ are purely imaginary for $z\in C_{2}$. It then follows from the Cauchy-Riemann equation that
\begin{equation}\label{eq:Rephi}
\left\{ \begin{array}{lll}
\Re  \phi_{\pm}(z)=0,  &z\in C_{2},\\
\Re \phi(z)>0, &z\in C_{1},\Im z>0,\\
\Re \widetilde{\phi}(z)>0, &z\in C_{1},\Im z<0,\\
\Re \phi(z)<0, &z\in \Sigma_E\cup\Sigma_I,\\
\Re \widetilde{\phi}_{\pm}(z)>0, &z\in (-\infty,-1),
  \end{array}
\right.
\end{equation}
where $\Sigma_E$ and $\Sigma_I$ are some arcs outside and inside of the unit circle as shown in Figure \ref{tauS2};
see also \cite[pages 1150-1151]{BDJ}.
\subsection{Global parametrix: $N$}
It is readily  seen from \eqref{eq:Rephi}
that on $\Sigma_I\cup\Sigma_E\cup C_{1}$ and bounded away from $z=-1$, the jump matrices for $S$ tend to the identity matrix exponentially fast as $n\to\infty$.
 Then, we arrive at the following approximate RH problem for $n$ large.
\begin{rhp}\label{RHP:NN}
We look for a $2 \times 2$ matrix-valued function $N(z)$ satisfying the following properties.
\begin{itemize}
\item[\rm (1)] $N(z)$ is analytic in $\mathbb{C} \setminus C_{2}$.
\item[\rm (2)] $N(z)$ has continuous boundary values on  $ C_{2}$ and they are related by 
  \begin{equation}\label{eq:Njump}
  N_{+}(z)=N_{-}(z)
 \begin{pmatrix}
  0& z^{\nu} \\
 - z^{-\nu} &0
  \end{pmatrix} ,\qquad z\in C_{2}.
  \end{equation}
  \item[\rm (3)] As $z\to \infty$, we have
  \begin{equation}\label{eq:N2Infinity}
  N(z)=I+O\left (\frac 1 {z}\right ).
 \end{equation}
\end{itemize}
\end{rhp}
A solution to the above RH problem can be constructed explicitly as follows
\begin{equation}\label{solu:N}
N(z)=D_{\infty}^{\sigma_{3}}X(z)D(z)^{-\sigma_{3}},
\end{equation}
where $X(z)$ is given by
\begin{equation}\label{eq:Nw}
X(z)=\frac{1}{2}\begin{pmatrix}\varrho+\varrho^{-1}&-i(\varrho-\varrho^{-1})\\
i(\varrho-\varrho^{-1})&\varrho+\varrho^{-1}\end{pmatrix},\quad \varrho=\varrho(z)=\left(\frac{z-\xi}{z-\xi^{-1}}\right)^{\frac{1}{4}}.
\end{equation} Here the branch for $\varrho(z)$ is chosen such that $\varrho(z)$ is analytic in $\mathbb{C} \setminus C_{2}$, behaves like $1$
as $z\rightarrow\infty$.
While  the Szeg\H{o} function $D(z)$ is defined as follows
\begin{equation}\label{Szego:D}
D(z)=(\varphi(z))^{\nu}, \quad \varphi(z)=\frac{z+1-\sqrt{(z-\xi)(z-\xi^{-1})}}{2\cos\frac{\theta_{c}}{2}},
\end{equation}
where $\xi=e^{i\theta_c}$, $z^{\nu}$ takes the principal branch and $\sqrt{(z-\xi)(z-\xi^{-1})}$ is analytic in $\mathbb{C} \setminus C_{2}$, and behaves like $z$
as $z\rightarrow\infty$. Therefore, $\varphi(z)$ is a conformal mapping from $\mathbb{C}\setminus C_{2}$ onto the inside of the unit circle with the center at $1/\cos\left(\frac{\theta_{c}}{2}\right)$. 
By \eqref{Szego:D}, it is easily seen that
\begin{equation}\label{Djumps}
D_{+}(z)D_{-}(z)=z^{\nu},
\end{equation}
for $z\in C_{2}$, and
\begin{equation}\label{Dinfty}
D_{\infty}=\lim_{z\rightarrow\infty}D(z)=\left(\cos\left(\frac{\theta_{c}}{2}\right)\right)^{\nu}.
\end{equation}
\subsection{Local parametrices: $P^{({\pm} )}$}
In this subsection, we seek two  parametrices $P^{({\pm} )}(z)$ satisfying the same jump condition as $S(z)$ in the neighborhoods $U(\xi^{\pm 1},\delta)$, and matching with $N(z)$ on the boundaries $\partial U(\xi^{\pm 1},\delta)$,  respectively. 
\begin{rhp}\label{RHP:P+}
\item[\rm (1)] $P^{({+})}(z)$ is analytic in $U(\xi,\delta)\setminus \Sigma_S$.
\item[\rm (2)] $P^{({+})}(z)$  satisfies the same  jump condition as $S(z)$ on $U(\xi,\delta)\cap \Sigma_S$, that is
\begin{equation}\label{eq:PJump}
  P^{({+})}_{+}(z)=P^{({+})}_{-}(z)
  \left\{ \begin{array}{lll}
\begin{pmatrix}1 & 0 \\ z^{-\nu} e^{2n\phi(z)} & 1\end{pmatrix} , & \hbox{for $z\in  \Sigma_E$}\cap U(\xi,\delta),\\[.5cm]
\begin{pmatrix}1&z^{\nu}e^{-2n\phi(z)}\\0 &1\end{pmatrix} , & \hbox{for $z\in C_{1}$} \cap U(\xi,\delta),\\[.5cm]
\begin{pmatrix}0& z^{\nu} \\-z^{-\nu} &0\end{pmatrix}, & \hbox{for $z\in C_{2}$} \cap U(\xi,\delta),\\[.5cm]
 \begin{pmatrix}1 & 0 \\z^{-\nu} e^{2n\phi(z)} &1\end{pmatrix}  ,   &   \hbox{for $z\in  \Sigma_I$}\cap U(\xi,\delta).\\
  \end{array}
\right.
  \end{equation}
   \item[\rm (3)] On the boundary of $U(\xi,\delta)$, $P^{({+})}(z)$ satisfies the matching condition
  \begin{equation}\label{eq:Matching+}
P^{({+})}(z)N^{-1}(z)= I+O(1/n).
 \end{equation}
\end{rhp}
To construct a solution to above RH problem, we first introduce the conformal mapping near $z=\xi$
\begin{equation}\label{eq:f}
f(z)=\left(\frac{3}{2}\phi(z)\right )^{\frac{2}{3}}\sim (\tau-1)^{\frac{1}{2}}e^{-i(\theta_{c}+
\frac{\pi}{2})}(z-\xi),\quad z \rightarrow \xi.
\end{equation}
Then, the solution to the above RH problem can be built out of the Airy function as follows:
\begin{equation}\label{solu:P+}
P^{({+})}(z)=E(z)\Phi^{\mathrm{(Ai)}}(n^{\frac{2}{3}}f(z))e^{n\phi(z)\sigma_{3}}z^{-\frac{\nu}{2}\sigma_{3}},
\end{equation}
where $\Phi^{\mathrm{(Ai)}}$ denotes the standard Airy parametrix given in  Appendix \ref{AP} and $E(z)$ is defined by
\begin{equation}\label{eq:E}
E(z)=N(z)z^{\frac{\nu}{2}\sigma_{3}}\frac{1}{\sqrt{2}}\begin{pmatrix}1&-i\\-i&1\end{pmatrix}(n^{\frac{2}{3}}f(z))^{\frac{\sigma_{3}}{4}}.
\end{equation}
The branches are chosen such that $\arg z\in(-\pi,\pi)$ and  $\arg f(z)\in(-\pi, \pi)$.
It follows from \eqref{eq:Njump} and \eqref{eq:E} that $E(z)$ is analytic in $U(\xi,\delta)$. Using \eqref{solu:N},  \eqref{solu:P+} and \eqref{AiryAsyatinfty}, we get the matching condition as $n\to\infty$
\begin{equation}\label{asy:P+N}
\begin{aligned}
&P^{({+})}(z)N(z)^{-1}\\
&=N(z)z^{\frac{\nu}{2}\sigma_{3}}\left(I+\frac{1}{48 n f(z)^{\frac{3}{2}}}
\begin{pmatrix}1&6i\\6i&-1\end{pmatrix}+O\left(\frac 1 {n^{2}}\right)\right)z^{-\frac{\nu}{2}\sigma_{3}}N(z)^{-1},\\
&=I+\frac{Q^{(+)}(z)}{n}+O\left(\frac 1 {n^{ 2}}\right),
\end{aligned}
\end{equation}
where the error term is uniform for $z\in\partial U(\xi,\delta) $.
The matrix $Q^{(+)}(z)$ is  defined by
\begin{equation}
Q^{(+)}(z)=\frac{1}{96f(z)^{\frac{3}{2}}}\begin{pmatrix}L_{1} &L_{2}\\ L_{3} & -L_{1}\end{pmatrix},
\end{equation}
where
\begin{equation*}
L_{1}=(3\alpha(z)+1) \varrho(z)^{2}-(3\alpha(z)-1) \varrho(z)^{-2},
\end{equation*}
\begin{equation*}
L_{2}=iD_{\infty}^2\left((1+3\alpha(z)) \varrho(z)^{2}+(3\alpha(z)-1) \varrho(z)^{-2}-6\beta(z)\right),
\end{equation*}
and
\begin{equation*}
L_{3}=iD_{\infty}^{-2}\left((1+3\alpha(z)) \varrho(z)^{2}+ (3\alpha(z)-1)\varrho(z)^{-2}+6\beta(z)\right),
\end{equation*}
with
\begin{equation}\label{eq:ab(z)}
\alpha(z)=D(z)^2z^{-\nu}+D(z)^{-2}z^{\nu}, \quad \beta(z)=D(z)^2z^{-\nu}-D(z)^{-2}z^{\nu}.
\end{equation}
It is seen from \eqref{Szego:D} and \eqref{Djumps} that $\alpha(z)$ and $\beta(z)/\sqrt{(z-\xi)(z-\xi^{-1})}$   are analytic near $z=\xi^{\pm 1}$. Moreover, we have  from \eqref{Szego:D} the following expansions
\begin{equation}\label{eq:abexpand1}
\alpha(z)=4i\nu^2(\tau-1)^{-\frac{1}{2}}e^{-i\theta_{c}}(z-\xi)+O((z-\xi)^{2}),\quad z\rightarrow \xi,
\end{equation}
\begin{equation}\label{eq:abexpand2}
\beta(z)=-4\nu(\tau-1)^{-\frac{1}{4}}e^{-i(\frac{\theta_{c}}{2}-\frac{\pi}{4})}(z-\xi)^{\frac{1}{2}}++O((z-\xi)^{\frac{3}{2}}),\quad z\rightarrow \xi,
\end{equation}
\begin{equation}\label{eq:abexpand3}
\alpha(z)=-4i\nu^2(\tau-1)^{\frac{1}{2}}e^{i\theta_{c}}(z-\xi^{-1})+O((z-\xi^{-1})^{2}),\quad z\rightarrow \xi^{-1},
\end{equation}
and
\begin{equation}\label{eq:abexpand4}
\beta(z)=-4\nu(\tau-1)^{\frac{1}{4}}e^{i(\frac{\theta_{c}}{2}-\frac{\pi}{4})}(z-\xi^{-1})^{\frac{1}{2}}+O((z-\xi^{-1})^{\frac{3}{2}}),\quad z\rightarrow \xi^{-1}.
\end{equation}
From \eqref{phi} and \eqref{eq:f}, we have
\begin{equation}\label{eq:fexpand}
f(z)^{-\frac{3}{2}}=\left(\frac{3}{2}\phi(z)\right)^{-1}=(z-\xi)^{-\frac{3}{2}}(\tau-1)^{-\frac{3}{4}}e^{i(\frac{3}{2}\theta_{c}+\frac{3}{4}\pi)}
+\kappa(z-\xi)^{-\frac{1}{2}}+O((z-\xi)^{\frac{1}{2}}),\quad z\rightarrow \xi,
\end{equation}
where
\begin{equation}\label{kappa}
\kappa=-\frac{3}{10}\tau^{\frac{1}{2}}(\tau-1)^{-\frac{5}{4}}e^{i(\theta_{c}+\frac{3}{4}\pi)}+\frac{6}{5}(\tau-1)^{-\frac{3}{4}}e^{i(\frac{1}{2}\theta_{c}+\frac{3}{4}\pi)}
-\frac{3\tau}{40}(\tau-1)^{-\frac{5}{4}}e^{i(\frac{3}{2}\theta_{c}+\frac{1}{4}\pi)}.
\end{equation}
From \eqref{eq:Nw}, we have
\begin{equation}\label{w2}
\varrho(z)^2=\frac{1}{2}\tau^{\frac{1}{2}}(\tau-1)^{-\frac{1}{4}}e^{-\frac{1}{4}\pi i}(z-\xi)^{\frac{1}{2}}\left(1+\frac{i\tau}{8(\tau-1)^{\frac{1}{2}}}(z-\xi)+O((z-\xi)^{2})\right), \quad z\rightarrow \xi,
\end{equation}
and
\begin{equation}\label{w-2}
\varrho(z)^{-2}=2\tau^{-\frac{1}{2}}(\tau-1)^{\frac{1}{4}}e^{\frac{1}{4}\pi i}(z-\xi)^{-\frac{1}{2}}+\frac{1}{4}\tau^{\frac{1}{2}}(\tau-1)^{-\frac{1}{4}}e^{-\frac{1}{4}\pi i}(z-\xi)^{\frac{1}{2}},\quad z\rightarrow \xi.
\end{equation}
Therefore, we obtain
\begin{equation}\label{Q+}
Q^{(+)}(z)=\frac{B_2}{(z-\xi)^2}+\frac{B_{1}}{z-\xi}+O(z-\xi),\quad z\rightarrow\xi,
\end{equation}
with
\begin{equation}\label{eq:B+}
B_2=\frac{5}{48}\tau^{-1/2}(\tau-1)^{-1/2}e^{i\frac{3}{2}\theta_{c}}\begin{pmatrix}1 &-iD_{\infty}^{2}\\ -iD_{\infty}^{-2} & -1\end{pmatrix},
\end{equation}
and
\begin{equation}\label{C+11}
\begin{aligned}
B_{1}&=\frac{1}{32}\tau^{-1/2}(\tau-1)^{-1}e^{i\frac{1}{2}\theta_{c}} \left(i\tau e^{i\theta_{c}}- \tau^{1/2} e^{i\frac{1}{2}\theta_{c}}+4(\tau-1)^{1/2}+8i\nu^2  \right)\begin{pmatrix}1 &0\\ 0& -1\end{pmatrix} \\
&+\frac{i}{96}\tau^{-1/2}(\tau-1)^{-1}e^{i\frac{1}{2}\theta_{c}} \left(4i\tau e^{i\theta_{c}}+3 \tau^{1/2} e^{i\frac{1}{2}\theta_{c}}-12(\tau-1)^{1/2}-24i\nu^2 \right)\begin{pmatrix}0 &D_{\infty}^{2}\\ D_{\infty}^{-2}& 0\end{pmatrix} \\
&+\frac{i}{4}\nu(\tau-1)^{-1}e^{i\theta_{c}}\begin{pmatrix}0 &-D_{\infty}^{2}\\ D_{\infty}^{-2}& 0\end{pmatrix}.
\end{aligned}
\end{equation}

Next,  we construct a similar  local parametrix  in the neighborhood of $\xi^{-1}$. The matrix-valued function $P^{({-})}(z)$ satisfies the following RH problem.
\begin{rhp}\label{PHP:P-}
\item[\rm (1)] $P^{({-})}(z)$ is analytic in $U(\xi^{-1},\delta)\setminus \Sigma_S$.
\item[\rm (2)] $P^{({-})}(z)$  satisfies the same  jump condition as $S(z)$ on $U(\xi^{-1},\delta)\cap \Sigma_S$.
\begin{equation}\label{eq:P-Jump}
  P^{({-})}_{+}(z)=P^{({-})}_{-}(z)
  \left\{ \begin{array}{lll}
\begin{pmatrix}1 & 0 \\ z^{-\nu} e^{2n\widetilde{\phi}(z)} & 1\end{pmatrix} , & \hbox{for $z\in  \Sigma_E$}\cap U(\xi^{-1},\delta),\\[.5cm]
\begin{pmatrix}1&z^{\nu}e^{-2n\widetilde{\phi}(z)}\\0 &1\end{pmatrix} , & \hbox{for $z\in C_{1}$} \cap U(\xi^{-1},\delta),\\[.5cm]
\begin{pmatrix}0& z^{\nu} \\-z^{-\nu} &0\end{pmatrix}, & \hbox{for $z\in C_{2}$} \cap U(\xi^{-1},\delta),\\[.5cm]
 \begin{pmatrix}1 & 0 \\z^{-\nu} e^{2n\widetilde{\phi}(z)} &1\end{pmatrix}  ,   &   \hbox{for $z\in  \Sigma_I$}\cap U(\xi^{-1},\delta).\\
  \end{array}
\right.
  \end{equation}
   \item[\rm (3)] On the boundary of $U(\xi^{-1},\delta)$, $P^{({-})}(z)$ satisfies the matching condition
  \begin{equation}\label{eq:Matching-}
P^{({-})}(z)N^{-1}(z)= I+O(1/n).
 \end{equation}
\end{rhp}
Similarly, the solution to the above RH problem can be expressed in terms of  the Airy function
\begin{equation}\label{solu:P-}
P^{({-})}(z)=\widetilde{E}(z)\Phi^{\mathrm{(Ai)}}(n^{\frac{2}{3}}\widetilde{f}(z))\sigma_{3}e^{n\widetilde{\phi}(z)\sigma_{3}}z^{-\frac{\nu}{2}\sigma_{3}},
\end{equation}
where
\begin{equation}\label{eq:tildef}
\widetilde{f}(z)=\left(\frac{3}{2}\widetilde{\phi}(z)\right )^{\frac{2}{3}}=(\tau-1)^{\frac{1}{2}}e^{i(\theta_{c}+\frac{\pi}{2})}(z-\xi^{-1}),\quad z \rightarrow \xi^{-1},
\end{equation}
and $\widetilde{E}(z)$ is defined as
\begin{equation}
\widetilde{E}(z)=N(z)z^{\frac{\nu}{2}\sigma_{3}}\sigma_{3}\frac{1}{\sqrt{2}}\begin{pmatrix}1&-i\\-i&1\end{pmatrix}(n^{\frac{2}{3}}\widetilde{f}(z))^{\frac{\sigma_{3}}{4}}.
\end{equation}
It is readily  seen that $\widetilde{E}(z)$ is analytic in $U(\xi^{-1},\delta)$. From \eqref{solu:N}, \eqref{solu:P-} and \eqref{AiryAsyatinfty}, we arrive at the uniform expansion for $z\in \partial U(\xi^{-1},\delta)$ as $n\to\infty$
\begin{equation}\label{asy:P-N}
\begin{aligned}
&P^{({-})}(z)N(z)^{-1}\\
&=N(z)z^{\frac{\nu}{2}\sigma_{3}}\sigma_{3}\left(I+\frac{1}{48\widetilde{f}(z)^{\frac{3}{2}}n}
\begin{pmatrix}1&6i\\6i&-1\end{pmatrix}+O\left(\frac 1 {n^{2}}\right)\right)\sigma_{3}z^{-\frac{\nu}{2}\sigma_{3}}N(z)^{-1},\\
&=I+\frac{Q^{(-)}(z)}{n}+O\left(\frac 1 {n^{2}}\right).
\end{aligned}
\end{equation}
Here, the matrix $Q^{(-)}(z)$ is given by
\begin{equation}
Q^{(-)}(z)=\frac{1}{96\widetilde{f}(z)^{\frac{3}{2}}}\begin{pmatrix}\widetilde{L}_{1} &\widetilde{L}_{2}\\ \widetilde{L}_{3} & -\widetilde{L}_{1}\end{pmatrix},
\end{equation}
with the entries
\begin{equation*}
\widetilde{L}_{1}=
(1-3\alpha(z)) \varrho(z)^{2}+(3\alpha(z)+1) \varrho(z)^{-2},
\end{equation*}

\begin{equation*}
\widetilde{L}_{2}=iD_{\infty}^{2}\left((1-3\alpha(z)) \varrho(z)^{2}-(3\alpha(z)+1) \varrho(z)^{-2}+6\beta(z)\right),
\end{equation*}
and
\begin{equation*}
\widetilde{L}_{3}=iD_{\infty}^{-2}\left((1-3\alpha(z)) \varrho(z)^{2}-(3\alpha(z)+1) \varrho(z)^{-2}-6\beta(z)\right);
\end{equation*}
see \eqref{eq:ab(z)} for the definition of $\alpha(z)$ and $\beta(z)$.
From \eqref{tildephi} and \eqref{eq:tildef}, we have the expansion as $z\rightarrow \xi^{-1}$
\begin{equation}\label{eq:tilde-f}
\widetilde{f}(z)^{-\frac{3}{2}}=\left(\frac{3}{2}\widetilde{\phi}(z)\right)^{-1}
=(z-\xi^{-1})^{-\frac{3}{2}}(\tau-1)^{-\frac{3}{4}}e^{-i(\frac{3}{2}\theta_{c}+\frac{3}{4}\pi)}
+\widetilde{\kappa}(z-\xi^{-1})^{-\frac{1}{2}}+O((z-\xi^{-1})^{\frac{1}{2}}),
\end{equation}
where
\begin{equation}\label{tilde-kappa}
\widetilde{\kappa}=-\frac{3}{10}\tau^{\frac{1}{2}}(\tau-1)^{-\frac{5}{4}}e^{-i(\theta_{c}+\frac{3}{4}\pi)}+\frac{6}{5}(\tau-1)^{-\frac{3}{4}}e^{-i(\frac{1}{2}\theta_{c}+\frac{3}{4}\pi)}
-\frac{3\tau}{40}(\tau-1)^{-\frac{5}{4}}e^{-i(\frac{3}{2}\theta_{c}+\frac{1}{4}\pi)}.
\end{equation}
From \eqref{eq:Nw}, we have
\begin{equation}\label{w2-1}
\varrho(z)^2=2\tau^{-\frac{1}{2}}(\tau-1)^{\frac{1}{4}}e^{-\frac{1}{4}\pi i}(z-\xi^{-1})^{-\frac{1}{2}}+\frac{1}{4}\tau^{1/2}(\tau-1)^{-1/4}e^{\frac{1}{4}\pi i}(z-\xi^{-1})^{\frac{1}{2}}+O(z-\xi^{-1})^{\frac{3}{2}},
\end{equation}
and
\begin{equation}\label{w-2-1}
\varrho(z)^{-2}=\frac{1}{2}\tau^{\frac{1}{2}}(\tau-1)^{-\frac{1}{4}}e^{\frac{1}{4}\pi i}(z-\xi^{-1})^{\frac{1}{2}}\left(1-\frac{i\tau}{8(\tau-1)^{\frac{1}{2}}}(z-\xi^{-1})+O((z-\xi^{-1})^{2})\right),
\end{equation}
as $z\rightarrow \xi^{-1}$.
Therefore, we have
\begin{equation}\label{Q-}
Q^{(-)}(z)=\frac{\widetilde{B}_2}{(z-\xi^{-1})^2}+  \frac{\widetilde{B}_{1}}{(z-\xi^{-1})}+O(z-\xi^{-1}),\quad z\rightarrow\xi^{-1},
\end{equation}
with
\begin{equation}\label{eq:B-}
\widetilde{B}_2=\frac{5}{48}\tau^{-1/2}(\tau-1)^{-1/2}e^{-i\frac{3}{2}\theta_{c}}\begin{pmatrix}1 &iD_{\infty}^{2}\\ iD_{\infty}^{-2} & -1\end{pmatrix},
\end{equation}
and 
\begin{equation}\label{C-11}
\begin{aligned}
\widetilde{B}_{1}&=\frac{1}{32}\tau^{-1/2}(\tau-1)^{-1}e^{-i\frac{1}{2}\theta_{c}} \left(-i\tau e^{-i\theta_{c}}- \tau^{1/2} e^{-i\frac{1}{2}\theta_{c}}+4(\tau-1)^{1/2}-8i\nu^2  \right)\begin{pmatrix}1 &0\\ 0& -1\end{pmatrix} \\
&+\frac{i}{96}\tau^{-1/2}(\tau-1)^{-1}e^{-i\frac{1}{2}\theta_{c}} \left(4i\tau e^{-i\theta_{c}}-3 \tau^{1/2} e^{-i\frac{1}{2}\theta_{c}}+12(\tau-1)^{1/2}-24i\nu^2 \right)\begin{pmatrix}0 &D_{\infty}^{2}\\ D_{\infty}^{-2}& 0\end{pmatrix} \\
&+\frac{i}{4}\nu(\tau-1)^{-1}e^{-i\theta_{c}}\begin{pmatrix}0 &D_{\infty}^{2}\\ -D_{\infty}^{-2}& 0\end{pmatrix}.
\end{aligned}
\end{equation}

\subsection{Final transformation: $S\to R$}
The final transformation is defined as follows
\begin{equation}\label{eq:RR}
  R(z)=\left\{ \begin{array}{ll}
 S(z)N(z)^{-1}, & \hbox{$z\in \mathbb{C}\backslash \{{U(\xi,\delta)\cup}U(\xi^{-1},\delta)\cup\Sigma_{S}$\},} \\
 S(z)P^{({+})}(z)^{-1}, & \hbox{$z\in U(\xi,\delta)\backslash \Sigma_{S}$,}\\
 S(z)P^{({-})}(z)^{-1},   & \hbox{$z\in U(\xi^{-1},\delta)\backslash \Sigma_{S}$.}             \end{array}
\right.
 \end{equation}
 Then, we have the following RH problem for $R(z)$.
\begin{rhp} \label{RHP: finalR}
$R(z)$ satisfies the following properties.
\item{(1)} $R(z)$ is analytic for $z\in \mathbb{C}\setminus\Sigma_{R}$, where the contour $\Sigma_{R}$ is illustrated in Figure \ref{tauR}.
\item{(2)} For  $z\in \Sigma_{R}$, we have $R_+(z)=R_-(z)J_{R}(z)$, where
 \begin{equation}\label{JumpRtilde}
 J_{R}(z)=\left\{\begin{aligned}
&P^{({+})}(z)N(z)^{-1},\quad &z&\in \partial U(\xi,\delta),\\
&P^{({-})}(z)N(z)^{-1},\quad &z&\in \partial U(\xi^{-1},\delta),\\
&N(z)J_{S}(z)N(z)^{-1},\quad &z&\in \Sigma_{R}\setminus(\partial U(\xi,\delta)\cup\partial U(\xi^{-1},\delta)).
\end{aligned}
\right.
 \end{equation}
\item{(3)} As $z\rightarrow\infty$, we have
\begin{equation}\label{Rtilde:infty}
R(z)=I+\frac{R_1}{z}+O\left(\frac 1 {z^{2}}\right).
\end{equation}
\end{rhp}

\begin{figure}[ht]
  \centering
  \includegraphics[width=10cm,height=5.2cm]{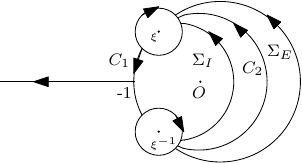}\\
  \caption{Contours and regions for the RH problem for $R(z)$}\label{tauR}
\end{figure}

From the matching conditions \eqref{eq:Matching+} and \eqref{eq:Matching-}, we have the following estimates as $n\rightarrow\infty$
\begin{equation}\label{JRtildeestimation}
J_{R}(z)=\left\{\begin{aligned}
&I+O\left(n^{-1}\right), &z&\in \partial U(\xi,\delta),\\
&I+O\left(n^{-1}\right), &z&\in \partial U(\xi^{-1},\delta),\\
&I+O\left(e^{-c_{2}n}\right),&z&\in \Sigma_{R}\setminus(\partial U(\xi,\delta)\cup\partial U(\xi^{-1},\delta)),
\end{aligned}\right.
\end{equation}
where $c_{2}$ is a positive constant. 
Therefore, $R$ satisfies a small-norm RH problem. By  the general theory for small-norm RH problems \cite[Section 7.2]{DKMVZ2} and \cite[Section 7.5]{D1}, we see that $R$ exists for $n$ large enough and satisfies the asymptotic approximation
\begin{equation}\label{Rtildeestimation}
R(z)=I+O\left(n^{-1}\right),\quad \mathrm{as}\quad n\rightarrow \infty,
\end{equation}
where the error term is uniform for $z\in \mathbb{C}\setminus\Sigma_{R}$.

\section{Proof of Theorem \ref{thm1}}\label{sec:proof1}
Tracing back the series of  invertible transformations $Y\to \widehat{Y}\to T\to S\to R$, we obtain
\begin{equation}\label{eq:YApprox1}
 Y(z)=e^{\frac{n}{2}\phi(z^{-})\sigma_{3}}n^{-\frac{ \nu}{2}\sigma_3} R(z)n^{\frac{ \nu}{2}\sigma_3} \left(\frac{z-z^{-}}{z}\right)^{\nu \sigma_3}e^{-\frac{n\tau}{2z}\sigma_3}z^{n\sigma_3}e^{-\frac{n}{2}\phi(z^{-})\sigma_{3}},
 \end{equation}
 for $|z|>1+\delta$ and
 \begin{equation}\label{eq:YApprox2}
 Y(z)=e^{\frac{n}{2}\phi(z^{-})\sigma_{3}}n^{-\frac{ \nu}{2}\sigma_3} R(z)n^{\frac{ \nu}{2}\sigma_3} (z-z^{-})^{\nu \sigma_3}i\sigma_2e^{-\frac{n\tau}{2}z\sigma_3}e^{\frac{n}{2}\phi(z^{-})\sigma_{3}},
 \end{equation}
  for  $|z|<\delta$, where
$
   \sigma_2=\begin{pmatrix}0 & -i\\ i & 0\end{pmatrix}
$
  is one of the Pauli matrices.

Thus, we get from \eqref{eq:REst}, \eqref{eq:A2Exp} and \eqref{eq:YApprox1} that 
\begin{equation}\label{eq:Y1}
 (Y_{-1})_{11}= -\frac{n}{2}\tau-\nu z^{-} +\frac{A^{(2)}_{11}}{n}+O(n^{-3/2}),
 \end{equation}
where
\begin{equation}
A^{(2)}_{11}=\frac{\nu^{2}}{z^{-}\lambda'(z^{-})^{2}}+\frac{3\nu^2\lambda''(z^{-})}{2\lambda'(z^{-})^{3}}
=-\frac{1}{2}\frac{\nu^2(z^{-})^{2}\tau}{(\tau+z^{-})^{2}},
\end{equation}
with $z^-$ defined in \eqref{eq:zpm} and $A^{(2)}$ given in \eqref{eq:A2Exp}.
Similarly,  from \eqref{eq:REst},  \eqref{eq:B2Exp},  \eqref{eq:A2Exp} and \eqref{eq:YApprox2} we obtain
\begin{equation}\label{eq:Y0n}
\begin{aligned}
\frac{d}{dz} \log (Y(z;n))_{21}|_{z=0}
&= -\frac{t}{2}+\frac{\nu}{z^{-}}+\frac{1}{n}\left(-\frac{A^{(2)}_{22}}{(z^{-})^{2}}
+\frac{2B^{(2)}_{22}}{(z^{-})^{3}}\right)+O(n^{-3/2}),
\end{aligned}
 \end{equation}
 where 
 \begin{equation}
-\frac{A^{(2)}_{22}}{(z^{-})^{2}}+\frac{2B^{(2)}_{22}}{(z^{-})^{3}}=-\frac{1}{2}\frac{\nu^2\tau}{(\tau+z^{-})^{2}}+\frac{\nu(\nu+1)}{(\tau+z^{-})}.
 \end{equation}
 By replacing $n$ and $\tau$ in \eqref{eq:Y0n} by $n+1$ and $\frac{n}{n+1}\tau$ respectively, we obtain
 \begin{equation}\label{eq:Y0}
\begin{aligned}
\frac{d}{dz} \log (Y(z;n+1))_{21}|_{z=0}
=& -\frac{t}{2}+ \frac{\nu}{z^{-}}+\frac{\nu }{n}\left(\frac{\tau }{(z^{-})^{2}\sqrt{1-\tau^2}}-\frac{1}{z^{-}}-\frac{1}{2}\frac{\nu\tau}{(\tau+z^{-})^{2}}+\frac{\nu+1}{(\tau+z^{-})}\right)\\
&+O(n^{-3/2}).
\end{aligned}
 \end{equation}

Similarly, we have from \eqref{eq:YApprox2} and \eqref{eq:REst} that
\begin{equation}\label{Asy:gammma}
Y_{12}(0;n)=\left(1+\frac{1}{n}\left(\frac{1}{2}
\frac{\nu^2z^{-}\tau}{(\tau+z^{-})^{2}}-
\frac{1}{2}\frac{\nu(\nu-1)z^{-}}{\tau+z^{-}}\right )\right)|z^{-}|^{\nu}
+O(n^{-3/2}),
\end{equation}
\begin{equation}\label{Asy:Y11}
Y_{11}(0;n)=-d_{1}\nu \lambda'(z^{-})^{-2\nu-1}|z^{-}|^{-2\nu-1}n^{-\nu-1/2}e^{n\phi(z^{-})}+O(n^{-\nu-1}),
\end{equation}
and
\begin{equation}\label{Asy:Y22}
Y_{22}(0;n)= d_{1}^{-1}\lambda'(z^{-})^{2\nu-1}|z^{-}|^{2\nu-1}n^{\nu-1/2}e^{-n\phi(z^{-})}+O(n^{\nu-1}),
\end{equation}
where $d_{1}=\frac{\sqrt{2 \pi}}{\Gamma(-\nu+1)}$.
By replacing $n$ and $\tau$ in \eqref{Asy:Y11} and \eqref{Asy:Y22} by $n+1$ and $\frac{n}{n+1}\tau$ respectively, we obtain
\begin{equation}\label{Asy:Y11n}
Y_{11}(0;n+1)=-d_{1}\nu\lambda'(z^{-})^{-2\nu-1}|z^{-}|^{-2\nu-1}n^{-\nu-1/2}e^{n\phi(z^{-})}z^{-}+O(n^{-\nu-1}),
\end{equation}
and
\begin{equation}\label{Asy:Y22n}
Y_{22}(0;n+1)= d_{1}^{-1}\lambda'(z^{-})^{2\nu-1}|z^{-}|^{2\nu-1}n^{\nu-1/2}e^{-n\phi(z^{-})}\frac{1}{z^{-}}+O(n^{\nu-1}).
\end{equation}

Substituting \eqref{eq:Y1}, \eqref{eq:Y0} and \eqref{Asy:gammma}-\eqref{Asy:Y22n}
into the differential identity \eqref{eq:dnuD},
we have
\begin{equation}\label{Asy:DnulogDn}
\begin{aligned}
\frac{d}{d\nu}\log D_{n,\nu}(n\tau)=&n\log(-z^{-})+\frac{n\tau}{2}(z^{-}-\frac{1}{z^{-}})-\nu\log n-\nu+\frac{1}{2}-\frac{\nu}{2}\log(1-\tau^2)\\
&-\nu\frac{d}{d\nu}\log\Gamma(1-\nu)+O\left(\frac{1}{n}\right),\quad n\rightarrow\infty,
\end{aligned}
\end{equation}
where the error term is uniform for $\tau$ in any compact subset of $(0,1)$ and the parameter $\nu$ in any compact subset of $\mathbb{C}\setminus \{1,2,3\cdots\}$.
Integrating  with respect to $\nu$ on both sides of \eqref{Asy:DnulogDn} yields
\begin{equation}\label{eq:DnAsy}
\begin{aligned}
\log D_{n,\nu}(n\tau)-\log D_{n,0}(n\tau)=&n\nu\log(-z^{-})+\frac{n\nu }{2}(z^{-}-\frac{1}{z^{-}})\tau-\frac{\nu^2}{2}\log n-\frac{\nu^2}{4}\log(1-\tau^2)\\
&-\frac{\nu^2}{2}+\frac{\nu}{2}-\nu\log\Gamma(1-\nu)+\int_{0}^{\nu}\log\Gamma(1-x)dx
+O(n^{-1})
,
\end{aligned}
\end{equation}
as $n\rightarrow\infty$.
Recalling the following integral representation of  the Barnes  $G$-function \cite[Eq. (5.17.4)]{Olver}
\begin{equation}\label{B-G}
\log G(z+1)=\frac{z}{2}\log 2\pi-\frac{z(z+1)}{2}+z\log\Gamma(z+1)-\int_{0}^{z}\log\Gamma(x+1)dx,
\end{equation}
  we further have
\begin{equation}\label{eq:ungappedDn}
\begin{aligned}
\log D_{n,\nu}(n\tau)-\log D_{n,0}(n\tau)=&n\nu\log(-z^{-})+\frac{n\nu }{2}(z^{-}-\frac{1}{z^{-}})\tau-\frac{\nu^2}{2}\log n-\frac{\nu^2}{4}\log(1-\tau^2)\\
&+\frac{\nu}{2}\log 2\pi+\log G(1-\nu)+O(n^{-1}),\quad n\rightarrow\infty
\end{aligned}
\end{equation}
for   $\nu\in\mathbb{C}\setminus \{1,2,3\cdots\}$.
From \cite[Lemma 7.1]{BDJ}, we have the asymptotic approximation of the logarithm of the Toeplitz determinant for the initial value $\nu=0$
\begin{equation}\label{eq:DnAsy0}
\log D_{n,0}(n\tau)=\frac{n^2}{4}\tau^2+O(e^{-cn}), \quad c >0,
\end{equation}
where the error term is uniform for $\tau$ in any compact subset of $(0,1)$.
Substituting \eqref{eq:DnAsy0} into \eqref{eq:ungappedDn}, we obtain \eqref{Integral:logDn1}.  The case of $\nu=1,2,3\cdots$ follows from
 \eqref{Integral:logDn1} and the symmetry 
\begin{equation}\label{eq: Symm}
\log D_{n,\nu}(t) =\log D_{n,-\nu}(t),\quad \nu\in\mathbb{N}.
\end{equation}
This completes the proof of Theorem \ref{thm1}.

\section{Proof of Theorem \ref{thm2}}\label{sec:proof2}
Tracing back the transformations $R\rightarrow S\rightarrow T\rightarrow \widehat{Y}\rightarrow Y$, we have for $|z-\xi^{\pm 1}|>\delta$
\begin{equation}\label{RtoY}
Y(z)=\left\{\begin{aligned}
&e^{-\frac{nl}{2}\sigma_{3}}R(z)N(z)e^{\frac{nl}{2}\sigma_{3}}e^{ng(z)\sigma_{3}},
&\mbox{for~even}~n,\\
&e^{-\frac{nl}{2}\sigma_{3}}\sigma_{3}R(z)N(z)\sigma_{3}e^{\frac{nl}{2}\sigma_{3}}
e^{ng(z)\sigma_{3}},&\mbox{for~odd}~n.
\end{aligned}
\right.
\end{equation}

To get the large-$z$ expansion of $N(z)$, we obtain from \eqref{solu:N}-\eqref{Szego:D} that
\begin{equation}\label{eq:X}
X(z)=\frac{1}{2}\begin{pmatrix}w+w^{-1}&-i(w-w^{-1})\\i(w-w^{-1})&w+w^{-1}\end{pmatrix}
=I-\frac{1}{4}(\xi-\xi^{-1})\begin{pmatrix}0&-i\\i&0\end{pmatrix}\frac{1}{z}
+O\left(\frac{1}{z^{2}}\right ),\quad z\rightarrow\infty,
\end{equation}
and
\begin{equation}\label{eq:szegoD}
D(z)^{-\sigma_{3}}=\left(I+\frac{1}{z}\frac{\nu}{\tau}\sigma_{3}
+O\left(\frac{1}{z^{2}}\right )\right)D_{\infty}^{-\sigma_{3}},\quad z\rightarrow\infty.
\end{equation}
Inserting \eqref{eq:X} and \eqref{eq:szegoD} into \eqref{solu:N}, we have the  large-$z$ asymptotics of $N(z)$
\begin{equation}\label{N:infty}
N(z)=I+\frac{N_{1}}{z}+O\left(\frac{1}{z^{2}}\right ), \quad z\rightarrow \infty,
\end{equation}
with
\begin{equation}\label{eq:N11}
(N_{1})_{11}=\frac{\nu}{\tau}.
\end{equation}
For later use, we also derive from \eqref{solu:N}-\eqref{Szego:D} the expansion of $N(z)$ as $z\to 0$
\begin{equation}\label{eq:N0}
N(z)=\begin{pmatrix}D_{\infty}^2(\tau-1)^{\frac{1}{2}}\tau^{-\frac{1}{2}}(1-(1+\nu)\tau^{-1}z+O(z^2))&\tau^{-\frac{1}{2}}+O(z)\\-\tau^{-\frac{1}{2}}(
1-((1+\nu)\tau^{-1}-1)z+O(z^2))&D_{\infty}^{-2}(\tau-1)^{\frac{1}{2}}\tau^{-\frac{1}{2}}+O(z)\end{pmatrix}. \end{equation}
By \eqref{g-function}, we have the following asymptotic expansion
\begin{equation}\label{G:infty}
e^{ng(z)\sigma_{3}}z^{-n\sigma_{3}}=I+\frac{G_{1}}{z}+O(z^{-2}),\quad z \rightarrow \infty.
\end{equation}
where
\begin{equation}
(G_{1})_{11}=-t\left(\tau^{-1}-\frac{1}{2}\tau^{-2}\right).%
\end{equation}
In virtue of \eqref{asy:P+N}, \eqref{asy:P-N}, \eqref{JumpRtilde} and \eqref{Rtilde:infty}, we have the following  asymptotic expansion 
\begin{equation}\label{error:R}
R(z)=I+\frac{R^{(1)}(z)}{n}+O(n^{-2}), \quad\text{for}~ z\in \partial U(\xi,\delta)\cup\partial U(\xi^{-1},\delta).
\end{equation}
Here the coefficient $R^{(1)}(z)$ behaves like $1/z$ as $z\rightarrow\infty$, and satisfies the jump relation
\begin{equation}
R^{(1)}_{+}(z)-R^{(1)}_{-}(z)=\left\{\begin{aligned}
&Q^{(+)}(z),\quad z\in \partial U(\xi,\delta),\\
&Q^{(-)}(z),\quad z\in \partial U(\xi^{-1},\delta),
\end{aligned}\right.
\end{equation}
with $Q^{(\pm)}(z)$ defined in \eqref{Q+} and \eqref{Q-}.
Therefore, we have from \eqref{Q+} and \eqref{Q-} that
\begin{equation}\label{eq:R1}
R^{(1)}(z)=\left\{\begin{aligned}
&\frac{B_{1}}{z-\xi}+\frac{\widetilde{B}_1}{z-\xi^{-1}}+\frac{B_{2}}{(z-\xi)^2}+\frac{\widetilde{B}_2}{(z-\xi^{-1})^2}, &z \notin \overline{U(\xi,\delta)}\cup\overline{U(\xi^{-1},\delta)},\\
&\frac{B_{1}}{z-\xi}+\frac{\widetilde{B}_1}{z-\xi^{-1}}+\frac{B_{2}}{(z-\xi)^2}+\frac{\widetilde{B}_2}{(z-\xi^{-1})^2}-Q^{({\pm})}(z), &z \in U(\xi,\delta)\cup U(\xi^{-1},\delta).
\end{aligned}\right.
\end{equation}
Expanding $R^{(1)}(z)$ as $z\rightarrow\infty$, we obtain from \eqref{C+11} and \eqref{C-11} that
\begin{equation}\label{R11tilde}
R(z)=I+\frac{R_1}{z}+O\left(\frac 1{z^2}\right ), \quad
\end{equation}
where 
\begin{equation}\label{eq:R111}
\left(R_1\right)_{11}=\frac{(B_{1}+\widetilde{B}_1)_{11}}{n}+O(n^{-2})=\frac{1}{8n}(1-4\nu^2)(\tau-1)^{-1}\tau^{-1}+O(n^{-2}).
\end{equation}
From \eqref{eq:R1}, and taking into account the expressions \eqref{eq:B+}, \eqref{C+11},  \eqref{eq:B-} and \eqref{C-11},  we obtain the following expansions  as $z\rightarrow0$
\begin{equation}\label{R210}
\begin{aligned}
R_{21}(z)=&\frac{(\tau-1)^{-\frac{1}{2}}}{24n}D_{\infty}^{-2} \left[\left(2+\tau^{-1}-12\nu^2\tau^{-1}\right)\right.
+\left(2-3\tau^{-1}-4\tau^{-2}(3\tau-4)\right.\\
&~\quad\left.\left.-12\nu^2\tau^{-2}(\tau-4)
-24\nu\tau^{-1}\right)z+O(z^2)\right]+O(n^{-2}),
\end{aligned}
\end{equation}
and
\begin{equation}\label{R220}
\begin{aligned}
R_{22}(z)=&1+\frac{1}{24n}\left(-3(\tau-1)^{-1}+\tau^{-1}
+12\nu^2\tau^{-1}(\tau-1)^{-1}\right)\\
&+\frac{1}{24n}\left[\left(3(\tau-1)^{-1}\tau^{-1}-4\tau^{-2}(\tau-4)\right.\right.\\
&~~\quad\quad\quad\left.\left.+12\nu^2\tau^{-2}(\tau-1)^{-1}(3\tau-4)\right)z+O(z^2)\right]
+O(n^{-2}).
\end{aligned}
\end{equation}


From \eqref{RtoY}, \eqref{N:infty}, \eqref{eq:N11}, \eqref{G:infty} and \eqref{eq:R111}, we have
\begin{align}\label{Y-1}
(Y_{-1})_{11}&=(R_{1})_{11}+(N_{1})_{11}+(G_{1})_{11}\\ \nonumber
&=-t(\tau^{-1}-\frac{1}{2}\tau^{-2})+\frac{\nu}{\tau}+\frac{1}{8n}(1-4\nu^2)(\tau-1)^{-1}\tau^{-1}+O(n^{-2}). 
\end{align}
From  \eqref{RtoY}, \eqref{g-function}, \eqref{eq:N0}, \eqref{R210} and  \eqref{R220}, we have
\begin{align}\label{eq:Y210}
\frac{d}{dz} \log (Y)_{21}(z,n+1)|_{z=0}=& (n+1)g'(0)+\frac{d}{dz} \log (R_{21}N_{11}+R_{22}N_{21})(z)|_{z=0}\nonumber\\
=&-t(\tau^{-1}-\frac{1}{2}\tau^{-2})-\frac{\nu}{\tau}+\frac{1}{8n}(\tau-1)^{-1}\tau^{-1}\nonumber\\
&-\frac{1}{2n}(\tau-1)^{-1}\tau^{-1}\nu^{2}+O(n^{-2}).
\end{align}
Similarly, we derive from  \eqref{RtoY}, \eqref{g-function}, \eqref{eq:N0} and \eqref{eq:R1} that
\begin{equation}\label{gamma-n}
Y_{12}(0;n)=e^{-nl}\tau^{-\frac{1}{2}}\left(1+\frac{1}{24n}(3(\tau-1)^{-1}+2-12\nu^{2}(\tau-1)^{-1})+O(n^{-2})\right),
\end{equation}
\begin{equation}\label{Y110}
Y_{11}(0;n)=e^{n\pi i}\left(1-\tau^{-1}\right)^{\nu+\frac{1}{2}}+O(n^{-2})\quad\mbox{and}\quad Y_{22}(0;n)=e^{-n\pi i}\left(1-\tau^{-1}\right)^{-\nu+\frac{1}{2}}+O(n^{-2}).
\end{equation}
By replacing $n$ and $\tau$ in \eqref{Y110} by $n+1$ and $\frac{n}{n+1}\tau$ respectively, we obtain
\begin{equation}\label{Y11n}
Y_{11}(0;n+1)=e^{(n+1)\pi i}\left(1-\tau^{-1}\right)^{\nu+\frac{1}{2}}\left(1-\frac{\nu+\frac{1}{2}}{n(\tau-1)}\right)+O(n^{-2}),
\end{equation}
and
\begin{equation}\label{Y22n}
Y_{22}(0;n+1)=e^{-(n+1)\pi i}\left(1-\tau^{-1}\right)^{-\nu+\frac{1}{2}}\left(1+\frac{\nu-\frac{1}{2}}{n(\tau-1)}\right)+O(n^{-2}).
\end{equation}

From the differential identity \eqref{eq:dnuD} and \eqref{Y-1}-\eqref{Y22n}, we have

\begin{equation}\label{Asy:DnulogDn1}
\begin{aligned}
\frac{d}{d\nu}\log D_{n,\nu}(t)=\nu\log(1-\tau^{-1})+O( {1}/{n}),
\end{aligned}
\end{equation}
where the error term is uniform for $\nu$ in any compact subset of the complex plane  and $\tau$ in any compact subset of $(1,+\infty)$.
Integrating \eqref{Asy:DnulogDn1} from 0 to $\nu$, we have
\begin{equation}\label{Integral:logDn}
\begin{aligned}
\log D_{n,\nu}(t)=\log D_{n,0}(t)+\frac{\nu^2}{2}\log(1-\tau^{-1})+O( {1}/{n}).
\end{aligned}
\end{equation}
The asymptotic approximation of the logarithm of the Toeplitz determinant with $\nu=0$ for $\tau>1$ is given by
\begin{equation}\label{eq:initial}
 \log D_{n,0}(\tau)=n^2\left(\tau-\frac{3}{4}-\frac{1}{2}\log \tau\right)-\frac{1}{12}\log n-\frac{1}{8}\log(1-\tau^{-1})+\zeta'(-1)+O( {1}/{n});
\end{equation}
see \cite[Eq. (2.10) ]{KO} and \cite{PS}.
Substituting \eqref{eq:initial} into \eqref{Integral:logDn} yields \eqref{Integral:logDn3}. This completes the proof of Theorem \ref{thm2}.
\begin{appendices}
\section{Model RH problems}
\subsection{Airy parametrix}\label{AP}
Let $\omega=e^{2\pi i/3}$, we define
\begin{equation}
\Phi^{(\mathrm{Ai})}(z)=M\left\{
\begin{aligned}
&\begin{pmatrix}
\mathrm{Ai}(z) & \mathrm{Ai}(\omega^{2}z) \\
\mathrm{Ai}^{\prime}(z) & w^{2} \mathrm{Ai}^{\prime}(\omega^{2}z)
\end{pmatrix} e^{-i \frac{\pi}{6} \sigma_{3}}, &z&\in \mathrm{I},\\
&\begin{pmatrix}
\mathrm{Ai}(z) & \mathrm{Ai}(\omega^{2}z) \\
\mathrm{Ai}^{\prime}(z) & \omega^{2} \mathrm{Ai}^{\prime}(\omega^{2}z)
\end{pmatrix} e^{-i \frac{\pi}{6} \sigma_{3}}\begin{pmatrix}
1 & 0 \\
-1 & 1
\end{pmatrix},&z&\in \mathrm{II},\\
&\begin{pmatrix}
\mathrm{Ai}(s) & -\omega^{2}\mathrm{Ai}(\omega z) \\
\mathrm{Ai}^{\prime}(z) & -\mathrm{Ai}^{\prime}(\omega z)
\end{pmatrix} e^{-i \frac{\pi}{6} \sigma_{3}}\begin{pmatrix}
1 & 0 \\
1 & 1
\end{pmatrix},&z&\in \mathrm{III},\\
&\begin{pmatrix}
\mathrm{Ai}(z) & -\omega^{2} \mathrm{Ai}(\omega z) \\
\mathrm{Ai}^{\prime}(z) & -\mathrm{Ai}^{\prime}(\omega z)
\end{pmatrix} e^{-i \frac{\pi}{6} \sigma_{3}},&z&\in \mathrm{IV},
\end{aligned}
\right.
\end{equation}
where $\mathrm{Ai}(z)$ is the Airy function (cf. \cite[Chapter 9]{Olver}),
$$
M=\sqrt{2\pi}e^{\frac{1}{6}\pi i}\begin{pmatrix}1 & 0\\ 0 & -i \end{pmatrix},
$$
and the regions I-IV are shown in Figure \ref{Airy4}. It is easy to check that $\Phi^{(\mathrm{Ai})}(z)$ solves the following RH problem (cf. \cite[Chapter 7]{D1}).

\subsection*{RH problem for $\Phi^{(\mathrm{Ai})}(z)$}

\begin{itemize}
\item[\rm (1)] $\Phi^\mathrm{(Ai)}(z)$ is analytic for $z\in \mathbb{C}\setminus \bigcup^4_{k=1}\Sigma_{k}$;
\item[\rm (2)] $\Phi^\mathrm{(Ai)}(z)$ satisfies the jump relations $\Phi^\mathrm{(Ai)}_+(z)=\Phi^\mathrm{(Ai)}_-(z)J^{(\mathrm{Ai})}_k(z)$, $z\in\Sigma_{k}$, $k=1,2,3,4$, where
\begin{equation*}
J^{(\mathrm{Ai})}_1(z)=\begin{pmatrix}1 & 1 \\ 0 & 1 \end{pmatrix},\
J^{(\mathrm{Ai})}_2(z)=\begin{pmatrix}1 & 0 \\ 1 & 1 \end{pmatrix}, \ J^{(\mathrm{Ai})}_3(z)=\begin{pmatrix}0 & 1 \\ -1 & 0 \end{pmatrix}, \ J^{(\mathrm{Ai})}_4(z)=\begin{pmatrix}1 & 0 \\ 1 & 1 \end{pmatrix}.
\end{equation*}
\item[\rm (3)] $\Phi^\mathrm{(Ai)}(z)$ satisfies the following asymptotic behavior as $z\rightarrow \infty$:
\begin{equation}\label{AiryAsyatinfty}
\Phi^{(\mathrm{Ai})}(z)=z^{-\frac{\sigma_{3}}{4} }\frac{1}{\sqrt{2}} \begin{pmatrix}1 & i\\ i &1\end{pmatrix}\left(I+\frac{1}{48z^{\frac{3}{2}}}\begin{pmatrix}1 & 6i \\ 6i & -1 \end{pmatrix}+O\left(z^{-3}\right)\right) e^{-\frac{2}{3} z^{\frac{3}{2}} \sigma_{3}}.
\end{equation}
\end{itemize}
\begin{figure}[ht]
  \centering
  \includegraphics[width=8cm,height=6cm]{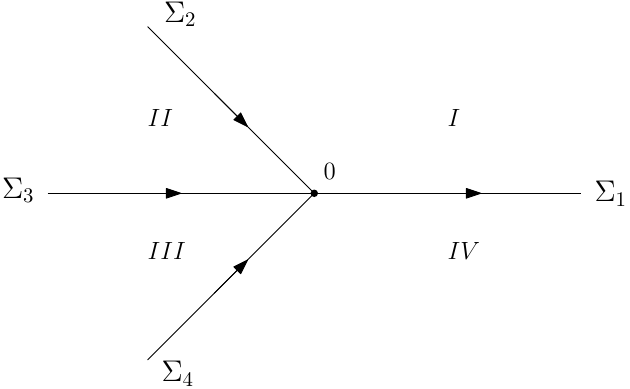}\\
  \caption{The jump contours and regions for $\Phi^{(\mathrm{Ai})}$}\label{Airy4}
\end{figure}

\subsection{Parabolic cylinder  parametrix}\label{PCP}
Define
\begin{equation}\label{eq:Parab1}Z_{0}(z)=2^{-\frac{\sigma_{3}}{2}}\begin{pmatrix}
D_{-\nu-1}(i z) & D_{\nu}(z) \\D_{-\nu-1}'(i z) & D_{\nu}'(z)\end{pmatrix}
\begin{pmatrix}
e^{i \frac{\pi}{2}(\nu+1)} & 0 \\
0 & 1
\end{pmatrix}
\end{equation}
and
\begin{equation}\label{eq:Parab2}Z_{n+1}(z)=Z_{n}(z) H_{n},\quad n=0,1,2,3,\end{equation}
where $D_{\nu}(z)$ denotes the standard parabolic cylinder function of order $\nu$ (cf. \cite[Chapter 12]{Olver}),
$$
H_{0}=\begin{pmatrix}1 & 0 \\ h_{0} & 1\end{pmatrix}, \
H_{1}=\begin{pmatrix}1 & h_{1} \\ 0 & 1\end{pmatrix}, \
H_{n+2}=e^{i \pi\left(\nu+\frac{1}{2}\right) \sigma_{3}} H_{n} e^{-i \pi\left(\nu+\frac{1}{2}\right) \sigma_{3}}, \ n=0,1,2
$$
with
\begin{equation}\label{h0}
h_{0}=-i \frac{\sqrt{2 \pi}}{\Gamma(\nu+1)}, \quad h_{1}=\frac{\sqrt{2 \pi}}{\Gamma(-\nu)} e^{i \pi \nu}, \quad 1+h_{0} h_{1}=e^{2 \pi i \nu}.
\end{equation}
\begin{figure}[h]
  \centering
  \includegraphics[width=6cm,height=5.5cm]{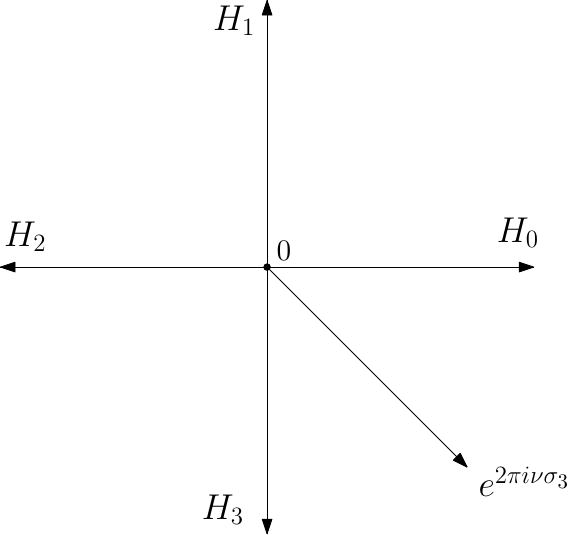}\\
  \caption{The jump contours and jump matrices for $\Phi^{(\mathrm{PC})}$}\label{PC}
\end{figure}

Let
$$
\Phi^{(\mathrm{PC})}(z)=\left\{\begin{aligned}
Z_{0}(z),\quad & \arg z \in\left(-\frac{\pi}{4}, 0\right), \\
Z_{1}(z),\quad & \arg z \in\left(0, \frac{\pi}{2}\right), \\
Z_{2}(z),\quad & \arg z \in\left(\frac{\pi}{2}, \pi\right), \\
Z_{3}(z),\quad & \arg z \in\left(\pi, \frac{3 \pi}{2}\right), \\
Z_{4}(z),\quad & \arg z \in\left(\frac{3 \pi}{2}, \frac{7 \pi}{4}\right),
\end{aligned}\right.
$$
then $\Phi^{(\mathrm{PC})}(z)$ solves the following RH problem (cf. \cite[Chapter 1.5]{FIKN}).

\begin{rhp}
\item[\rm (1)] $\Phi^\mathrm{(PC)}(z)$ is analytic for all $z\in \mathbb{C}\setminus \bigcup^5_{k=1}\Sigma_{k}$, where $\Sigma_{k}=\{z\in\mathbb{C}:\arg z=\frac{k\pi}{2}\}$, $k=1,2,3,4$ and  $\Sigma_{5}=\{z\in\mathbb{C}:\arg z=-\frac{\pi}{4}\}$; see Figure \ref{PC}.
\item[\rm (2)] $\Phi^\mathrm{(PC)}(z)$ satisfies the jump conditions as shown in Figure \ref{PC}.
\item[\rm (3)] 
As $z\to\infty$, we have
\begin{align}\label{PCAsyatinfty}
\Phi^\mathrm{(PC)}(z)=\begin{pmatrix}0 &1 \\ 1 & -z\end{pmatrix}2^{\frac{\sigma_3}{2}}\begin{pmatrix}
1+\frac{\nu(\nu+1)}{2z^2}+O\left(\frac{1}{z^4}\right) & \frac{\nu}{z}+O\left(\frac{1}{z^3}\right) \\ \frac{1}{z}+O\left(\frac{1}{z^3}\right) & 1-\frac{\nu(\nu-1)}{2z^2}+O\left(\frac{1}{z^4}\right)\end{pmatrix}
z^{-\nu\sigma_3} e^{\frac{z^{2}}{4}\sigma_{3}},
\end{align} 
where the branch for $z^{-\nu}$ is chosen such that $\arg z\in (- \frac{ \pi}{4},  \frac{7 \pi}{4})$.
\end{rhp}

\end{appendices}

\section*{Acknowledgements}
The work of Shuai-Xia Xu was supported in part by the National Natural Science Foundation of China under grant numbers 11971492 and 12371257, and by  Guangdong Basic and Applied Basic Research Foundation (Grant No.2022B1515020063). Yu-Qiu Zhao was supported in part by the National Natural Science Foundation of China under grant numbers 11971489 and 12371077.

\end{document}